\documentclass[3p,11pt]{elsarticle}

\newcommand{\revisedA}[1]{#1\xspace}
\newcommand{\markRevisedA}{}

\usepackage{latexsym}
\usepackage{xspace}
\usepackage{amssymb}
\usepackage{amsmath}
\usepackage{stmaryrd}
\usepackage{hyperref}
\usepackage{url}
\usepackage{graphicx}
\input{xy}
\xyoption{all}
\usepackage{calc}
\usepackage{tikz}
\usepackage{flushend}
\usepackage{color}
\usepackage{framed}

\newtheorem{theorem}{Theorem}[section]
\newtheorem{lemma}[theorem]{Lemma}
\newtheorem{proposition}[theorem]{Proposition}
\newtheorem{corollary}[theorem]{Corollary}

\newtheorem{Definition}[theorem]{Definition}
\newtheorem{Example}[theorem]{Example}
\newtheorem{Remark}[theorem]{Remark}

\newenvironment{example}{\begin{Example}\begin{em}}{\end{em}\end{Example}}
\newenvironment{remark}{\begin{Remark}\begin{em}}{\end{em}\end{Remark}}
\newproof{proof}{Proof}

\def\eqref#1{(\ref{#1})}
\def\defeq{\stackrel{\mathrm{def}}{=}}
\def\tuple#1{\langle#1\rangle}

\newcommand{\E}{\exists}

\newcommand{\mL}{\mathcal{L}}

\newcommand{\mR}{\mathcal{R}}

\newcommand{\mF}{\mathcal{F}}

\newcommand{\mM}{\mathcal{M}}
\newcommand{\mMp}{{\mathcal{M}'\!}}

\newcommand{\NN}{\mathbb{N}}

\newcommand{\myend}{\mbox{}\hfill{\footnotesize$\Box$}}

\newcommand{\comment}[1]{}

\newcommand{\fand}{\varotimes}

\newcommand{\fto}{\Rightarrow}
\newcommand{\fequiv}{\Leftrightarrow}

\newcommand{\SP}{\Sigma_P}
\newcommand{\SA}{\Sigma_A}

\newcommand{\fPDLx}{$\mathit{fPDL}$\xspace}
\newcommand{\fKx}{$\mathit{fK}$\xspace}
\newcommand{\fPDL}{$\mathit{fPDL}_{\scriptscriptstyle \revisedA{\fand}}$\xspace}
\newcommand{\fK}{$\mathit{fK}_{\!\scriptscriptstyle \revisedA{\fand}}$\xspace}

\newcommand{\DeltaM}{\Delta^{\!\mathcal{M}}}
\newcommand{\DeltaMp}{\Delta^{\!\mathcal{M}'\!}}

\newcommand{\Vmin}{V_{min}}
\newcommand{\Vmax}{V_{max}}

\newcommand{\mN}{\mathcal{N}}
\newcommand{\Mz}{M_0}

\journal{arXiv}

\begin{document}
\sloppy
	
\begin{frontmatter}
		
\title{Fuzzy directed simulations for fuzzy modal logics over residuated lattices}

\author{Linh Anh Nguyen}
\ead{nguyen@mimuw.edu.pl}
\ead{nalinh@ntt.edu.vn}				
\address{Institute of Informatics, University of Warsaw, Banacha 2, 02-097 Warsaw, Poland}
\address{Faculty of Information Technology, Nguyen Tat Thanh University, Ho Chi Minh City, Vietnam}

\begin{abstract}
We introduce the notion of fuzzy directed simulation between fuzzy Kripke models over \revisedA{linear and complete residuated lattices} and investigate its fundamental properties. In particular, we prove that all positive formulas of \revisedA{a fuzzy extension of propositional dynamic logic} are preserved under fuzzy directed simulations and establish a Hennessy--Milner theorem for this notion.
Furthermore, we present a method for computing the greatest fuzzy directed simulation between two finite fuzzy Kripke models and implement it for the case where the underlying residuated lattice is the G\"odel, product, or \L{}ukasiewicz structure. Finally, we experimentally evaluate the performance of the implementation and present the obtained results.
\end{abstract}


\end{frontmatter}


\section{Introduction}
\label{section:intro}

Bisimulation and simulation are well-known notions in computer science. Bisimulations are used to characterize the equivalence of states in Kripke models \cite{vanBenthemCorr,vBenthem76,vBenthem83} and transition systems \cite{HennessyM85,Park81}. From a logical perspective, modal formulas are invariant under bisimulations, and, given two image-finite Kripke models, the greatest relation between their domains under which all modal formulas are invariant coincides with the greatest bisimulation between the models. The latter property is called the Hennessy--Milner property of bisimulations. Simulations are used to compare the observable behavior of states in transition systems and automata \cite{Park81,fac/He89}. From a logical perspective, positive existential modal formulas are preserved under simulations, and, given two image-finite Kripke models, the greatest relation between their domains that preserves all positive existential modal formulas coincides with the greatest simulation between the models~\cite{BRV2001}. The notion of directed simulation is closely related to (bi)simulation and has been studied in the context of modal and description logics \cite{KurtoninaR97,BSDL-P-LOGCOM}. Positive modal formulas are preserved under directed simulations, and, given two image-finite Kripke models, the greatest relation between their domains that preserves all positive modal formulas is precisely the greatest directed simulation between the models.

To deal with vagueness and imprecision, fuzzy systems and fuzzy logics are often used as generalizations of their crisp counterparts. Accordingly, the notions of bisimulation, simulation, and directed simulation have been extended to the fuzzy setting, although one notable gap remains. These extensions follow two main approaches, depending on whether the relations themselves are required to be crisp or are allowed to take fuzzy truth values. From the logical perspective, if the underlying modal logic includes involutive negation and/or the Baaz projection operator, then the greatest fuzzy bisimulation or simulation between two image-finite Kripke models or transition systems is necessarily crisp~\cite{Fan15,FSS2020,DBLP:conf/fuzzIEEE/NguyenN21,DBLP:journals/tfs/NguyenN23,DBLP:journals/cas/NguyenN22}. 

Crisp bisimulations have been investigated for fuzzy transition systems~\cite{CaoCK11,CaoSWC13,DBLP:journals/fss/WuD16,DBLP:journals/ijar/WuCHC18,DBLP:journals/fss/WuCBD18}, weighted automata~\cite{DamljanovicCI14}, fuzzy modal logics~\cite{EleftheriouKN12,Fan15}, and fuzzy description logics~\cite{FSS2020,DBLP:journals/tfs/NguyenN23}. Fuzzy bisimulations have been studied for fuzzy automata~\cite{CiricIDB12,CiricIJD12,ijar/Nguyen23}, weighted and fuzzy social networks~\cite{ai/FanL14,IgnjatovicCS15}, fuzzy modal logics~\cite{Fan15,FBSML}, and fuzzy description logics~\cite{FSS2020,TFS2020}. Similarly, crisp simulations have been considered for fuzzy transition systems~\cite{DBLP:journals/ijar/PanLC15,fss/WuD16,DBLP:conf/fuzzIEEE/NguyenN21,DBLP:journals/jifs/Nguyen22} and weighted automata~\cite{DamljanovicCI14}, whereas fuzzy simulations have been investigated for fuzzy transition systems~\cite{DBLP:journals/ijar/PanC0C14,DBLP:journals/ijar/PanLC15}, fuzzy automata~\cite{CiricIDB12,CiricIJD12,TFS2020,ijar/Nguyen23}, fuzzy social networks~\cite{IgnjatovicCS15}, fuzzy modal logics~\cite{DBLP:journals/cas/NguyenN22}, and fuzzy description logics~\cite{TFS2020}. Crisp directed simulations have also been studied for fuzzy transition systems~\cite{DBLP:conf/fuzzIEEE/NguyenN21}. 
For further details on the logical characterization and computation of bisimulations, simulations, and directed simulations for fuzzy systems, we refer the reader to the related work sections of~\cite{FBSML,DBLP:journals/tfs/NguyenN23,DBLP:journals/cas/NguyenN22,DBLP:journals/jifs/Nguyen22}.

\subsection{Motivation}

To the best of our knowledge, fuzzy directed simulations have not yet been investigated. This gap motivates the present work. First, introducing this notion completes the picture of bisimulations, simulations, and directed simulations for fuzzy systems. Second, since fuzzy bisimulations have been successfully applied to the minimization of fuzzy interpretations in fuzzy description logics~\cite{minimization-by-fBS,corr/abs-2510-21423}, fuzzy directed simulations may provide a suitable foundation for studying the minimization of fuzzy Kripke models while preserving positive modal formulas, as well as the minimization of fuzzy interpretations in fuzzy description logics while preserving positive concept assertions.

\subsection{Our contributions}

In this work, we introduce the notion of fuzzy directed simulation between fuzzy Kripke models over \revisedA{linear and complete residuated lattices} and investigate its fundamental properties. In particular, we prove that all positive formulas of \revisedA{a fuzzy extension of propositional dynamic logic} are preserved under fuzzy directed simulations and establish a Hennessy--Milner theorem for this notion.

Furthermore, we present a method for computing the greatest fuzzy directed simulation between two finite fuzzy Kripke models. It is based on a reduction to the problem of computing the greatest correct marking of a finite fuzzy minimax net~\cite{DBLP:journals/tfs/NguyenMS23}.

Finally, we implement the method for the case where the underlying residuated lattice is the G\"odel, product, or \L{}ukasiewicz structure. We experimentally evaluate its performance and report the obtained results.

\subsection{The structure of this work}

The remainder of this work is organized as follows. Section~\ref{section: prel} introduces the necessary preliminaries. Section~\ref{section: fbir} defines fuzzy directed simulations between fuzzy Kripke models and studies their properties. Section~\ref{section: computation} presents the method for computing the greatest fuzzy directed simulation between two finite fuzzy Kripke models. Section~\ref{section: impl and performance} describes the implementation and reports the results of the performance evaluation. Section~\ref{section: conc} concludes the work. \ref{section: appendix} contains the proofs of the theoretical results.


\section{Preliminaries}
\label{section: prel}

This section introduces the necessary background on residuated lattices, fuzzy sets, fuzzy relations, fuzzy modal logics, and fuzzy minimax nets. The material is adapted from~\cite{FBSML,DBLP:journals/tfs/NguyenMS23}.

\subsection{Residuated lattices and fuzzy sets}

A {\em residuated lattice} \cite{Hajek1998,Belohlavek2002} is an algebra $\mL = \tuple{L, \leq, \fand, \fto, 0, 1}$ such that 
\begin{itemize}
\item $\tuple{L, \leq, 0, 1}$ is a lattice with the least element 0 and the greatest element 1,
\item $\tuple{L, \fand, 1}$ is a commutative monoid, 
\item for every $x, y, z \in L$,
\begin{equation}
x \fand y \leq z \ \ \textrm{iff}\ \ x \leq (y \fto z). \label{fop: GDJSK 00} 
\end{equation}
\end{itemize}

Given a residuated lattice \mbox{$\mL = \tuple{L, \leq, \fand, \fto, 0, 1}$}, let $\land$ and $\lor$ denote the {\em meet} and {\em join} operators associated with the lattice. By $x \fequiv y$ we denote \mbox{$(x \fto y) \land (y \fto x)$}. The operations $\fto$ and $\fequiv$ are called {\em residuum} and {\em biresiduum}, respectively. We use the convention that $\fand$ and $\land$ bind stronger than $\lor$, which in turn binds stronger than $\fto$ and $\fequiv$. 

A residuated lattice $\mL = \tuple{L, \leq, \fand, \fto, 0, 1}$ is {\em linear} (respectively, {\em complete}) if the bounded lattice \mbox{$\tuple{L, \leq, 0, 1}$} is linear (respectively, complete). 
Popular linear and complete residuated lattices include the G\"odel, product, and {\L}ukasiewicz structures, with $L = [0,1]$ and $\fand$ being a t-norm. They are specified below:
\begin{center}
	\begin{tabular}{|c|c|c|c|}
		\hline
		& G\"odel & Product & {\L}ukasiewicz \\
		\hline
		$x \fand y$ & $\min\{x,y\}$ & $x \cdot y$ & $\max\{x+y-1, 0\}$ \\
		\hline
		$x \fto y$ 
		& 
		\(
		\left\{
		\!\!\!\begin{array}{ll}
		y & \textrm{if $x > y$} \\ 
		1 & \textrm{otherwise}
		\end{array}\!\!\!
		\right.
		\)
		& 
		\(
		\left\{
		\!\!\!\begin{array}{ll}
		y/x & \textrm{if $x > y$} \\ 
		1 & \textrm{otherwise}
		\end{array}\!\!\!
		\right.
		\)	
		& $\min\{y - x + 1, 1\}$
		\\ \hline
	\end{tabular}
\end{center}

From now on, let $\mL = \tuple{L, \leq, \fand, \fto, 0, 1}$ be an arbitrary linear and complete residuated lattice. 

We say that the multiplication $\fand$ is {\em right-continuous} if, for every $x \in L$ and $Y \subseteq L$, 
\[ x \fand \bigwedge\!Y = \bigwedge_{y \in Y} (x \fand y). \]
Clearly, all the G\"odel, product, and {\L}ukasiewicz t-norms are right-continuous. 

\begin{lemma}[cf.~\cite{FBSML,Belohlavek2002,Hajek1998}]\label{lemma: JHFJWx}
The following properties hold for all $a,a',b,b',c \in L$:
\begin{eqnarray}
(1 \fto a) & = & a \label{eq: JHFJWx 0} \\
(a \fto (b \fto c)) & = & (b \fto (a \fto c)) \label{eq: JHFJWx 1} \\
(a \fto (b \fto c)) & \leq & (a \fand b \fto c) \label{eq: JHFJWx 10}  \\
a \fand (b \fto c) & \leq & ((a \fto b) \fto c) \label{eq: JHFJWx 20} \\
a \fand (b \fto c) & \leq & (b \fto a \fand c) \label{eq: JHFJWx 30} \\
(a \fto a') \land (b \fto b') & \leq & (a \land b \fto a' \land b') \label{eq: JHFJWx 40} \\
(a \fto a') \land (b \fto b') & \leq & (a \lor b \fto a' \lor b') \label{eq: JHFJWx 50} \\
(a \fto b) & \leq & ((c \fto a) \fto (c \fto b)) \label{eq: JHFJWx 60} \\
\revisedA{(a \fto b)} & \revisedA{\leq} & \revisedA{(c \fand a \fto c \fand b)} \label{eq: JHFJWx 65} \\
\revisedA{(a \fequiv b)} & \revisedA{\leq} & \revisedA{(c \fand a \fequiv c \fand b)} \label{eq: JHFJWx 66}
\end{eqnarray}
\end{lemma}

The equality \eqref{eq: JHFJWx 0} can be proved straightforwardly using~\eqref{fop: GDJSK 00} and the fact that 1 is the neutral element w.r.t.~$\fand$. 
The equality \eqref{eq: JHFJWx 1} and the inequalities \eqref{eq: JHFJWx 10}  and \eqref{eq: JHFJWx 20} come directly from \cite{FBSML}, while the inequalities \eqref{eq: JHFJWx 30}--\eqref{eq: JHFJWx 60} can be proved analogously to \cite[inequalities~(10) and (16)--(18)]{FBSML}, respectively. 
\revisedA{The inequality~\eqref{eq: JHFJWx 65} can easily be proved using~\eqref{fop: GDJSK 00} and the associativity and monotonicity of~$\fand$. The inequality~\eqref{eq: JHFJWx 66} directly follows from~\eqref{eq: JHFJWx 65}.} 

Let $X$ be a non-empty set. Any mapping $f: X \to L$ is called a {\em fuzzy subset} of $X$, or simply a {\em fuzzy set}. It is {\em empty} if $f(x) = 0$ for all $x \in X$. 
For elements $x_1,\ldots,x_n \in X$ and values $a_1,\ldots,a_n \in L$, we use the notation $\{x_1:a_1$, \ldots, $x_n:a_n\}$ to represent the fuzzy subset $f$ of $X$ defined by $f(x_i) = a_i$ for each $1 \leq i \leq n$, and $f(x) = 0$ for all $x \in X \setminus \{x_1,\ldots,x_n\}$. 
Given fuzzy subsets $f$ and $g$ of $X$, we write $f \leq g$ if $f(x) \leq g(x)$ for all $x \in X$. 

A fuzzy subset of $X \times Y$ is called a {\em fuzzy relation} between $X$ and $Y$. A fuzzy relation between $X$ and itself is called a fuzzy relation on~$X$. The {\em converse} of a fuzzy relation $R: X \times Y \to L$ is \mbox{$R^- : Y \times X \to L$} defined by $R^-(y,x) = R(x,y)$, for every $\tuple{y,x} \in Y \times X$. The {\em composition} of fuzzy relations \mbox{$R: X \times Y \to L$} and \mbox{$S: Y \times Z \to L$} is the fuzzy relation $R \circ S : X \times Z \to L$ such that, for every $x \in X$ and $z \in Z$, 
\[ (R \circ S)(x,z) = \bigvee_{y \in Y} (R(x,y) \fand S(y,z)). \] 

Given a set $\mR$ of fuzzy relations between $X$ and $Y$, by $\bigvee\!\mR$ we denote the fuzzy relation between $X$ and $Y$ such that, for every $x \in X$ and $y \in Y$, 
\[ (\bigvee\!\mR)(x,y) = \bigvee_{R \in \mR}\! R(x,y). \] 

A fuzzy relation $R$ on $X$ is 
\begin{itemize}
\item {\em reflexive} if $R(x,x) = 1$ for all $x \in X$, 
\item {\em transitive} if $R(x,y) \fand R(y,z) \leq R(x,z)$ for all $x,y,z \in X$. 
\end{itemize}
It is a {\em fuzzy pre-order} if it is reflexive and transitive. 

\subsection{Fuzzy modal logics}

Let $\SA$ denote a non-empty set of {\em actions}, which are also called {\em atomic programs}, 
and let $\SP$ denote a non-empty set of {\em propositions}, which are also called {\em atomic formulas}. 
By \fPDL we denote a fuzzy extension of propositional dynamic logic. {\em Programs} and {\em formulas} of \fPDL over $\mL$ and a signature $\tuple{\SA,\SP}$ are defined inductively as follows:
\begin{itemize}
\item all $\varrho \in \SA$ are programs of \fPDL, 
\item if $\alpha$ and $\beta$ are programs of \fPDL and $\varphi$ is a formula of \fPDL, then $\alpha \circ \beta$, $\alpha \cup \beta$, $\alpha^*$, and $\varphi?$ are programs of \fPDL, 
\item all $a \in L$ and $p \in \SP$ are formulas of \fPDL, 
\item if $\varphi$ and $\psi$ are formulas of \fPDL, $\alpha$ is a program of \fPDL, and $a \in L$, then \revisedA{$a \fand \varphi$,} $\varphi \land \psi$, $\varphi \lor \psi$, $\varphi \to \psi$, $[\alpha]\varphi$, and $\tuple{\alpha}\varphi$ are formulas of \fPDL.\footnote{One can define the residual negation by \mbox{$\lnot\varphi \defeq \varphi \to 0$}.}
\end{itemize} 

We use symbols such as $\varrho$, $\alpha$, $\beta$, $\varphi$, $\psi$, $a$, $p$ with the meanings as in the above definition.

By \fK we denote the largest fragment of \fPDL that disallows program constructors. Thus, the only programs in \fK are actions from $\SA$. 
\revisedA{By \fPDLx and \fKx we denote the restrictions of \fPDL and \fK, respectively, that disallow the formula constructor \mbox{$a \fand \varphi$}.}

A {\em positive formula} (respectively, {\em program}) of \fPDL is a formula (respectively, program) of \fPDL that does not use the test constructor~$\varphi?$ and in which implications \mbox{$\varphi \to \psi$} are restricted to the form $a \to \psi$, where $a \in L$. Similarly, a {\em positive formula} of \fK is a formula of \fK in which implications \mbox{$\varphi \to \psi$} are restricted to the form $a \to \psi$, where $a \in L$. \revisedA{The notions of a positive formula of \fPDLx or \fKx are defined analogously.}

For $\Phi = \{\varphi_1,\ldots,\varphi_n\}$, we denote 
\begin{eqnarray*}
\textstyle\bigvee\!\Phi & = & \varphi_1 \lor\ldots\lor \varphi_n \lor 0 \\
\textstyle\bigwedge\!\Phi & = & \varphi_1 \land\ldots\land \varphi_n \land 1 \\
\textstyle\bigotimes\!\Phi & = & \varphi_1 \fand\ldots\fand \varphi_n \fand 1.
\end{eqnarray*}

\begin{table}\centering
\fbox{\parbox{0.7\linewidth}{\centering
\(
\begin{array}{rcl}
(\alpha \circ \beta)^\mM(x,y) & \!=\! & \bigvee\{\alpha^\mM(x,z) \fand \beta^\mM(z,y) \mid z \in \Delta^\mM \} \\
(\alpha \cup \beta)^\mM(x,y) & \!=\! & \alpha^\mM(x,y) \lor \beta^\mM(x,y) \\
(\alpha^*)^\mM(x,y) & \!=\! & \bigvee \{\textstyle\bigotimes\{\alpha^\mM(x_i,x_{i+1}) \mid 0 \leq i < n\} \mid \\
	&& \qquad n \geq 0,\ x_0,\ldots,x_n \in \Delta^\mM, x_0 = x,\ x_n = y\} \\
(\varphi?)^\mM(x,y) & \!=\! & \textrm{($\varphi^\mM(x)$ if $x = y$ else 0)} \\[0.5ex]
a^\mM(x) & \!=\! & a \\
\revisedA{(a \fand \varphi)^\mM(x)} & \!\revisedA{=}\! & \revisedA{a \fand \varphi^\mM(x)} \\
(\varphi \land \psi)^\mM(x) & \!=\! & \varphi^\mM(x) \land \psi^\mM(x) \\
(\varphi \lor \psi)^\mM(x) & \!=\! & \varphi^\mM(x) \lor \psi^\mM(x) \\
(\varphi \to \psi)^\mM(x) & \!=\! & (\varphi^\mM(x) \fto \psi^\mM(x)) \\
([\alpha]\varphi)^\mM(x) & \!=\! & \bigwedge \{\alpha^\mM(x,y) \fto \varphi^\mM(y) \mid y \in \Delta^\mM\} \\
(\tuple{\alpha}\varphi)^\mM(x) & \!=\! & \bigvee \{\alpha^\mM(x,y) \fand \varphi^\mM(y) \mid y \in \Delta^\mM\}.
\end{array}
\)
}}
\caption{The interpretation of complex programs and formulas\label{table: HFHAD}.}
\end{table}

A {\em fuzzy Kripke model} over $\mL$ and $\tuple{\SA,\SP}$ is a pair $\mM = \tuple{\Delta^\mM, \cdot^\mM}$, where $\DeltaM$ is a non-empty set, called the {\em domain}, and $\cdot^\mM$ is the {\em interpretation function} that maps each $\varrho \in \SA$ to a fuzzy relation $\varrho^\mM$ on $\DeltaM$ and maps each $p \in \SP$ to a fuzzy subset $p^\mM$ of $\DeltaM$.
The interpretation function is extended to complex programs and formulas as shown in Table~\ref{table: HFHAD}.

A fuzzy Kripke model $\mM$ is {\em image-finite} if the set $\{y \in \Delta^\mM \mid \varrho^\mM(x,y) > 0\}$ is finite, for every $x \in \Delta^\mM$ and $\varrho \in \SA$. It is {\em finite} if $\Delta^\mM$ \revisedA{and the signature $\tuple{\SA,\SP}$ are} finite. 
$\mM$ is {\em witnessed} w.r.t.\ \fPDL if every infinite set under the infimum (respectively, supremum) operator in Table~\ref{table: HFHAD} has a smallest (respectively, greatest) element, when considering positive formulas and programs of \fPDL. The notion of being witnessed w.r.t.\ \fK is defined analogously.  
Observe that every finite fuzzy Kripke model is witnessed w.r.t.\ \fPDL, and every image-finite fuzzy Kripke model is witnessed w.r.t.~\fK. If $L$ is finite, then all fuzzy Kripke models are witnessed w.r.t.\ both \fPDL and~\fK. 


\subsection{Fuzzy minimax nets}
\label{sec: minimax}

A {\em (fuzzy) minimax net}~\cite{DBLP:journals/tfs/NguyenMS23} over $\mL$ is a structure $\mN = \tuple{\Vmin, \Vmax, E, \Mz}$, where $\Vmin$ and $\Vmax$ are non-empty disjoint sets of {\em nodes}, consisting of {\em min-nodes} and {\em max-nodes}, respectively, $E: (\Vmin \times \Vmax) \cup (\Vmax \times \Vmin) \to L$ is a fuzzy set of {\em edges}, and $\Mz: \Vmin \to L$ is the {\em marking limit for min-nodes}. 
%
%
A {\em positive} edge of $\mN$ is a pair $\tuple{x,y} \in (\Vmin \times \Vmax) \cup (\Vmax \times \Vmin)$ with $E(x,y) > 0$. 
By $|E|$ we denote the {\em size} of $E$, which is defined to be the number of positive edges of~$\mN$. 

Given a minimax net $\mN = \tuple{\Vmin, \Vmax, E, \Mz}$, a fuzzy set $M: \Vmin \cup \Vmax \to L$ is called a {\em marking} of $\mN$. It is a {\em correct marking} of $\mN$ if, for every $x \in \Vmin$ and $y \in \Vmax$, 
\begin{eqnarray}
	M(x) & \leq & \Mz(x) \label{eq: JHDSD 1} \\
	M(x) & \leq & \bigwedge_{z \in \Vmax}\!\! (E(z,x) \fto M(z)) \label{eq: JHDSD 2} \\
	M(y) & \leq & \bigvee_{z \in \Vmin}\!\!(E(z,y) \fand M(z)). \label{eq: JHDSD 3} 
\end{eqnarray}
A marking $M$ of $\mN$ is {\em stable} if, for every $x \in \Vmin$ and $y \in \Vmax$, 
\begin{eqnarray}
	M(x) & = & \Mz(x) \land \bigwedge_{z \in \Vmax}\!\!(E(z,x) \fto M(z)) \label{eq: JHDSD 4} \\
	M(y) & = & \bigvee_{z \in \Vmin}\!\!(E(z,y) \fand M(z)). \label{eq: JHDSD 5} 
\end{eqnarray}

Note that, by definition, stable markings are correct. 
It has been proved in~\cite{DBLP:journals/tfs/NguyenMS23} that every minimax net has the greatest correct marking and that the greatest correct marking of a minimax net is stable.


\section{Fuzzy directed simulations between fuzzy Kripke models}
\label{section: fbir}

In this section, we introduce fuzzy directed simulations between fuzzy Kripke models and present their properties. In particular, we provide a theorem on preservation of positive formulas of \fPDL under fuzzy directed simulations, as well as a theorem establishing the Hennessy--Milner property for fuzzy directed simulations. The proofs of the results in this section are provided in~\ref{section: appendix}.

Given fuzzy Kripke models $\mM$ and $\mMp$, a fuzzy relation \mbox{$Z : \Delta^\mM \times \Delta^\mMp \to L$} is called a {\em fuzzy directed simulation} between $\mM$ and $\mMp$ if the following conditions hold for all $p \in \SP$, $\varrho \in \SA$ and all possible values for the free variables:
\begin{eqnarray}
&& Z(x,x') \leq (p^\mM(x) \fto p^\mMp(x')) \label{eq: FDS1} \\
&& \E y' \in \Delta^\mMp\ (Z(x,x') \fand \varrho^\mM(x,y) \leq \varrho^\mMp(x',y') \fand Z(y,y')) \label{eq: FDS2} \\
&& \E y \in \Delta^\mM\ (Z(x,x') \fand \varrho^\mMp(x',y') \leq \varrho^\mM(x,y) \fand Z(y,y')). \label{eq: FDS3}
\end{eqnarray}

The above notion differs from the notion of a fuzzy bisimulation~\cite{FBSML} in that~\eqref{eq: FDS1} uses~$\fto$ instead of~$\fequiv$. 
It differs from the notion of a fuzzy simulation~\cite{DBLP:journals/cas/NguyenN22} in that the condition~\eqref{eq: FDS3} is included. 

A fuzzy directed simulation between $\mM$ and itself is called a {\em fuzzy directed auto-simulation} of~$\mM$.

\begin{figure}[t]
\begin{center}
\begin{tabular}{|c|}
\hline
\begin{tikzpicture}[->,>=stealth,auto,black]
	\node (M) {$\mM$};
	\node (Mp) [node distance=6cm, right of=M] {$\mMp$};
	\node (u) [node distance=1cm, below of=M] {$u\!:\!p_{\,1.0}$};
	\node (up) [node distance=1cm, below of=Mp] {$u'\!:\!p_{\,1.0}$};
	\node (bu) [node distance=2.5cm, below of=u] {};
	\node (vp)  [node distance=2.5cm, below of=up] {$v'\!:\!q_{\,0.5}$};
	\node (v) [node distance=1.5cm, left of=bu] {$v\!:\!q_{\,0.7}$};
	\node (w) [node distance=1.5cm, right of=bu] {$w\!:\!q_{\,0.4}$};
	\draw (u) to node [left]{\footnotesize{0.7}} (v);	
	\draw (u) to node [right]{\footnotesize{0.9}} (w);	
    \draw (v) edge[bend left=15] node[above]{\footnotesize{0.6}} (w);
    \draw (w) edge[bend left=15] node[below]{\footnotesize{0.8}} (v);
	\draw (up) to node [right]{\footnotesize{0.8}} (vp);	
    \draw (vp) edge[out=-125,in=-55,looseness=7] node[below]{\footnotesize{0.7}} (vp);
\end{tikzpicture}
\\
\hline
\end{tabular}
\caption{Illustration of the fuzzy Kripke models used in Example~\ref{example: JHRJA}.\label{fig: HFJSK}}
\end{center}
\end{figure}
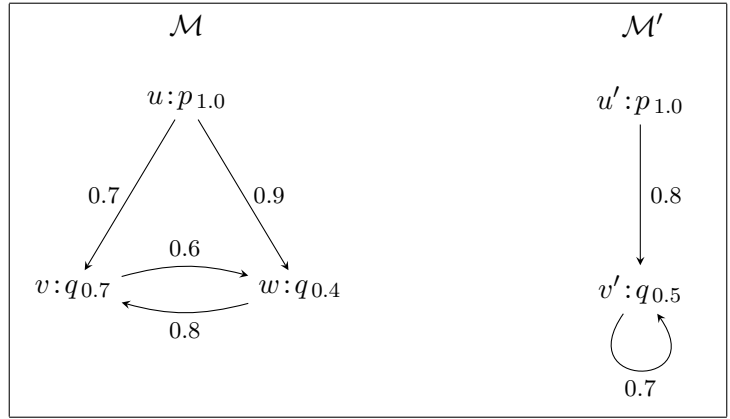

\begin{example}\label{example: JHRJA}
Let $\SP = \{p,q\}$ and $\SA = \{\sigma\}$. Let $\mM$ and $\mMp$ be the fuzzy Kripke models illustrated in Figure~\ref{fig: HFJSK} and specified as follows:
\begin{itemize}
\item $\DeltaM = \{u,v,w\}$,\ \ $\DeltaMp = \{u',v'\}$, 
\item $p^\mM = \{u\!:\!1\}$, $q^\mM = \{v\!:\!0.7, w\!:\!0.4\}$,\ \ $p^\mMp = \{u'\!:\!1\}$, $q^\mMp = \{v'\!:\!0.5\}$, 
\item $\sigma^\mM = \{\tuple{u,v}\!:\!0.7, \tuple{u,w}\!:\!0.9, \tuple{v,w}\!:\!0.6, \tuple{w,v}\!:\!0.8\}$,\ \ 
      $\sigma^\mMp = \{\tuple{u',v'}\!:\!0.8, \tuple{v',v'}\!:\!0.7\}$.
\end{itemize}

We now determine the greatest fuzzy directed simulation $Z$ between $\mM$ and $\mMp$ when the underlying residuated lattice is the G\"odel, product, or {\L}ukasiewicz structure. Due to~\eqref{eq: FDS1} (for $p,q \in \SP$), we have $Z(u,v') = 0$. Moreover, when the underlying residuated lattice is the G\"odel or product structure, we also have $Z(v,u') = Z(w,u') = 0$.

\begin{itemize}
\item \textbf{Case of the G\"odel structure:} Since $q^\mM(v) = 0.7$ and $q^\mMp(v') = 0.5$, by~\eqref{eq: FDS1} with $q$ in place of $p$, we have $Z(v,v') \leq (0.7 \fto 0.5) = 0.5$. By~\eqref{eq: FDS2}, we have
\begin{eqnarray*}
&& Z(w, v') \leq (\varrho^\mM(w,v) \fto \varrho^\mMp(v',v') \fand Z(v,v')) \leq (0.8 \fto 0.7 \fand 0.5) = 0.5, \\[0.5ex]
&& Z(u, u') \leq (\varrho^\mM(u,v) \fto \varrho^\mMp(u',v') \fand Z(v,v')) \leq (0.7 \fto 0.8 \fand 0.5) = 0.5.
\end{eqnarray*}
It is straightforward to verify that
$Z = \{\tuple{u,u'}\!:\!0.5, \tuple{v,v'}\!:\!0.5, \tuple{w,v'}\!:\!0.5\}$
is a fuzzy directed simulation between $\mM$ and $\mMp$, and hence the greatest one.

\item \textbf{Case of the product structure:} Let $a = Z(v,v')$. By~\eqref{eq: FDS2} and~\eqref{eq: FDS3}, we have
\begin{eqnarray*}
&& Z(w, v') \leq (\varrho^\mM(w,v) \fto \varrho^\mMp(v',v') \fand Z(v,v')) \leq (0.8 \fto 0.7 \fand a) = \tfrac{7a}{8}, \\[0.5ex]
&& Z(v,v') \leq (\varrho^\mMp(v',v') \fto \varrho^\mM(v,w) \fand Z(w,v')) \leq (0.7 \fto 0.6 \fand \tfrac{7a}{8}) = \tfrac{3a}{4}.
\end{eqnarray*}
Thus, $a \leq \tfrac{3a}{4}$, which implies $a = 0$, and hence $Z(v,v') = 0$. Consequently, by~\eqref{eq: FDS2}, it follows that $Z(w, v') \fand \varrho^\mM(w,v) = 0$ and $Z(u, u') \fand \varrho^\mM(u,v) = 0$, and therefore $Z(w, v') = Z(u,u') = 0$. That is, $Z$ is the empty fuzzy relation.

\item \textbf{Case of the {\L}ukasiewicz structure:} Since $q^\mM(v) = 0.7$ and $q^\mMp(v') = 0.5$, by~\eqref{eq: FDS1} with $q$ in place of $p$, we have $Z(v,v') \leq (0.7 \fto 0.5) = 0.8$. As in the previous case, applying~\eqref{eq: FDS2} and~\eqref{eq: FDS3}, we obtain
\begin{eqnarray*}
&& Z(w, v') \leq (0.8 \fto 0.7 \fand 0.8) = 0.7, \\[0.5ex]
&& Z(v,v') \leq (0.7 \fto 0.6 \fand 0.7) = 0.6, \\[0.5ex]
&& Z(w, v') \leq (0.8 \fto 0.7 \fand 0.6) = 0.5, \\[0.5ex]
&& Z(v,v') \leq (0.7 \fto 0.6 \fand 0.5) = 0.4,
\end{eqnarray*}
\begin{eqnarray*}
&& Z(w, v') \leq (0.8 \fto 0.7 \fand 0.4) = 0.3, \\[0.5ex]
&& Z(v,v') \leq (0.7 \fto 0.6 \fand 0.3) = 0.3, \\[0.5ex]
&& Z(w, v') \leq (0.8 \fto 0.7 \fand 0.3) = 0.2.
\end{eqnarray*}
By~\eqref{eq: FDS3}, we further obtain  
\[ Z(v, u') \leq (\varrho^\mMp(u',v') \fto \varrho^\mM(v,w) \fand Z(w, v')) \leq (0.8 \fto 0.6 \fand 0.2) = 0.2. \]
By~\eqref{eq: FDS2}, we have 
\begin{eqnarray*}
&& Z(w, u') \leq (\varrho^\mM(w,v) \fto \varrho^\mMp(u',v') \fand Z(v, v')) \leq (0.8 \fto 0.8 \fand 0.3) = 0.3, \\
&& Z(u, u') \leq (\varrho^\mM(u,w) \fto \varrho^\mMp(u',v') \fand Z(w, v')) \leq (0.9 \fto 0.8 \fand 0.2) = 0.1.
\end{eqnarray*}
Thus, 
\[ Z \leq \{\tuple{u, u'}\!:\!0.1, \tuple{v, u'}\!:\!0.2, \tuple{v, v'}\!:\!0.3, \tuple{w, u'}\!:\!0.3, \tuple{w, v'}\!:\!0.2\}. \] 
It is straightforward to verify that the fuzzy relation on the right-hand side is a fuzzy directed simulation between $\mM$ and $\mMp$, and therefore it is the greatest one (i.e., it coincides with $Z$). 
\end{itemize}

As an example illustrating that directed simulation differs from simulation and bisimulation, it can be verified (e.g., by using our codes~\cite{FDSML-prog} together with the input files ``g5.in'' and ``g5p.in'') that the greatest fuzzy simulation and the greatest fuzzy bisimulation between $\mM$ and $\mM'$ are as follows, respectively:
\begin{itemize}
\item Case of the G\"odel structure:
\begin{eqnarray*}
& \{\tuple{u,u'}\!:\!0.5, \tuple{v,v'}\!:\!0.5, \tuple{w,v'}\!:\!0.5\},\\
& \{\tuple{u,u'}\!:\!0.4, \tuple{v,v'}\!:\!0.4, \tuple{w,v'}\!:\!0.4\};
\end{eqnarray*}

\item Case of the product structure:
\begin{eqnarray*}
& \{\tuple{u, u'}\!:\!\tfrac{5}{9}, \tuple{v, v'}\!:\!\tfrac{5}{7}, \tuple{w, v'}\!:\!0.625\}, \\
& \emptyset \textrm{ (the empty fuzzy relation)};
\end{eqnarray*}

\item Case of the \L{}ukasiewicz structure:
\begin{eqnarray*}
& \{\tuple{u, u'}\!:\!0.6, \tuple{v, u'}\!:\!0.3, \tuple{v, v'}\!:\!0.8, \tuple{w, u'}\!:\!0.6, \tuple{w, v'}\!:\!0.7\},\\
& \{\tuple{u, u'}\!:\!0.1, \tuple{v, v'}\!:\!0.3, \tuple{w, v'}\!:\!0.2\}.
\end{eqnarray*}
\end{itemize}

\vspace{-2em}
\myend
\end{example}

\begin{proposition}\label{prop: HFHSJ}
Let $\mM$, $\mMp$ and $\mM''$ be image-finite fuzzy Kripke models.
\begin{enumerate}
\item\label{ass: HFHSJ 1} The fuzzy relation $Z : \Delta^\mM \times \Delta^\mM \to L$ specified by $Z(x,x')$ = (1 if $x = x'$ else 0) is a fuzzy directed auto-simulation of~$\mM$.
\item\label{ass: HFHSJ 3} If $Z_1$ is a fuzzy directed simulation between $\mM$ and $\mMp$, and $Z_2$ is a fuzzy directed simulation between $\mMp$ and $\mM''$, then $Z_1 \circ Z_2$ is a fuzzy directed simulation between $\mM$ and $\mM''$.
\item\label{ass: HFHSJ 4} If $\mF$ is a family of fuzzy directed simulations between $\mM$ and $\mMp$, then $\bigvee\!\mF$ is also a fuzzy directed simulation between $\mM$ and $\mMp$.
\end{enumerate}   
\end{proposition}

\newcommand{\ProofPropositionHFHS}{
\begin{proof}[of Proposition~\ref{prop: HFHSJ}]
The first assertion clearly holds. The second assertion can be proved analogously to the third assertion of~\cite[Proposition~3.6]{FBSML}, using the following fact~\cite[Lemma~2.1]{FBSML}: 
\[ (a \fto b) \fand (b \fto c) \leq (a \fto c). \]

Consider the third assertion, with $\mF$ being a set of fuzzy directed simulations between~$\mM$ and~$\mMp$. 
The condition~\eqref{eq: FDS1} with $Z$ replaced by $\bigvee\!\mF$ clearly holds. 

Consider the condition~\eqref{eq: FDS2} with $Z$ replaced by $\bigvee\!\mF$. Let $x, y \in \DeltaM$, $x' \in \DeltaMp$, and $\varrho \in \SA$. It is well known that $\fand$ is left-continuous. Hence, 
\[
   \big(\bigvee\!\mF\big)(x,x') \fand \varrho^\mM(x,y) = \big(\!\bigvee_{Z \in \mF}\! Z(x,x')\big) \fand \varrho^\mM(x,y) = \bigvee_{Z \in \mF} (Z(x,x') \fand \varrho^\mM(x,y)).
\]
For each $Z \in \mF$, since $Z$ is a fuzzy directed simulation between $\mM$ and $\mMp$, there exists $y'_Z \in \DeltaMp$ such that 
\begin{equation*}
Z(x,x') \fand \varrho^\mM(x,y) \leq \varrho^\mMp(x', y'_Z) \fand Z(y, y'_Z).
\end{equation*} 
Since $\mMp$ is image-finite and $\mL$ is linear, it follows that there exists a common $y' \in \DeltaMp$ such that
\begin{equation*}
\bigvee_{Z \in \mF}\!(\varrho^\mMp(x', y'_Z) \fand Z(y, y'_Z)) \leq \bigvee_{Z \in \mF}\!(\varrho^\mMp(x', y') \fand Z(y, y')).
\end{equation*} 
With that $y'$ we have 
\[
         \big(\bigvee\!\mF\big)(x,x') \fand \varrho^\mM(x,y) 
    \leq \bigvee_{Z \in \mF} (\varrho^\mMp(x', y') \fand Z(y, y')) 
    = \varrho^\mMp(x', y') \fand \bigvee_{Z \in \mF}\! Z(y, y').
\]
Therefore, 
\[
         \big(\bigvee\!\mF\big)(x,x') \fand \varrho^\mM(x,y) 
    \leq \varrho^\mMp(x', y') \fand \big(\bigvee\!\mF\big)(y, y'),
\]
which implies that the condition~\eqref{eq: FDS2} with $Z$ replaced by $\bigvee\!\mF$ holds.

The condition~\eqref{eq: FDS3} with $Z$ replaced by $\bigvee\!\mF$ can be proved analogously.
\myend
\end{proof}
} 

The following corollary follows directly from the above proposition.

\begin{corollary}\label{cor: HFHSJ}
Let $\mM$ be an image-finite fuzzy Kripke model. 
Then the greatest fuzzy directed auto-simulation of $\mM$ exists and is a fuzzy pre-order.
\end{corollary}


We say that a formula $\varphi$ is {\em preserved} under fuzzy directed simulations if, 
for every pair of fuzzy Kripke models $\mM$ and $\mM'$ that are witnessed w.r.t.\ \fPDL and for every fuzzy directed simulation $Z$ between $\mM$ and $ \mM'$, \mbox{$Z(x,x') \leq (\varphi^\mM(x) \fto \varphi^\mMp(x'))$} for all $x \in \Delta^\mM$ and $x' \in \Delta^\mMp$. 
 
\begin{theorem} \label{theorem: preservation}
All positive formulas of \fPDL are preserved under fuzzy directed simulations.
\end{theorem}

This theorem is an immediate consequence of the following lemma.

\begin{lemma} \label{lemma: GDHAW}
Let $\mM$ and $\mM'$ be fuzzy Kripke models that are witnessed w.r.t.\ \fPDL and let $Z$ be a fuzzy directed simulation between $\mM$ and $\mM'$. Then the following properties hold for every positive formula $\varphi$ of \fPDL, every positive program $\alpha$ of \fPDL, and all values of the free variables:
\begin{eqnarray}
&& Z(x,x') \leq (\varphi^\mM(x) \fto \varphi^\mMp(x')) \label{eq: GDHAW 1} \\[0.5ex]
&& \E y' \in \Delta^\mMp\ (Z(x,x') \fand \alpha^\mM(x,y) \leq \alpha^\mMp(x',y') \fand Z(y,y')) \label{eq: GDHAW 2} \\
&& \E y \in \Delta^\mM\ (Z(x,x') \fand \alpha^\mMp(x',y') \leq \alpha^\mM(x,y) \fand Z(y,y')). \label{eq: GDHAW 3} 
\end{eqnarray}
\end{lemma}

\newcommand{\ProofLemmaGDHAW}{
\begin{proof}[of Lemma~\ref{lemma: GDHAW}]
This proof is adapted from the proof of~\cite[Lemma~4.2]{FBSML} (on invariance of fuzzy bisimulations). We proceed by induction on the structure of $\varphi$ and~$\alpha$.

	Consider the assertion~\eqref{eq: GDHAW 1}. 
	The case where $\varphi = a$ is trivial. 
	The case where $\varphi = p$ follows from the condition~\eqref{eq: FDS1}. 
	\begin{itemize}
		\item {\markRevisedA Case $\varphi = a \fand \psi$: We have $\varphi^\mM(x) = a \fand \psi^\mM(x)$ and \mbox{$\varphi^\mMp(x') = a \fand \psi^\mMp(x')$}. By the induction assumption of~\eqref{eq: GDHAW 1}, 
		\begin{equation}\label{eq: JFHSN 1}
		Z(x,x') \leq \left(\psi^\mM(x) \fto \psi^\mMp(x')\right).
		\end{equation}
		By~\eqref{eq: JHFJWx 65}, 
		\begin{eqnarray}
		&& \left(\psi^\mM(x) \fto \psi^\mMp(x')\right) \leq  \left(a \fand \psi^\mM(x) \fto a \fand \psi^\mMp(x')\right). \label{eq: JFHSN 2}
		\end{eqnarray}
		The assertion~\eqref{eq: GDHAW 1} follows from~\eqref{eq: JFHSN 1} and \eqref{eq: JFHSN 2}.
        } 
		
		\item Case $\varphi = \psi \land \xi$: We have $\varphi^\mM(x) = \psi^\mM(x) \land \xi^\mM(x)$ and \mbox{$\varphi^\mMp(x') = \psi^\mMp(x') \land \xi^\mMp(x')$}. By the induction assumption of~\eqref{eq: GDHAW 1}, 
		\begin{eqnarray}
		Z(x,x') & \leq & \left(\psi^\mM(x) \fto \psi^\mMp(x')\right) \label{eq: JDGSH 1} \\
		Z(x,x') & \leq & \left(\xi^\mM(x) \fto \xi^\mMp(x')\right). \label{eq: JDGSH 2}
		\end{eqnarray}
		By~\eqref{eq: JHFJWx 40}, 
		\begin{eqnarray}
		&& \left(\psi^\mM(x) \fto \psi^\mMp(x')\right) \land \left(\xi^\mM(x) \fto \xi^\mMp(x')\right) \leq \left(\varphi^\mM(x) \fto \varphi^\mMp(x')\right). \label{eq: JDGSH 3}
		\end{eqnarray}
		The assertion~\eqref{eq: GDHAW 1} follows from~\eqref{eq: JDGSH 1}, \eqref{eq: JDGSH 2} and~\eqref{eq: JDGSH 3}. 
		
		\item The case $\varphi = (\psi \lor \xi)$ is similar to the previous case, using~\eqref{eq: JHFJWx 50} instead of~\eqref{eq: JHFJWx 40}.  
		
		\item Case $\varphi = (a \to \psi)$: We have $\varphi^\mM(x) = (a \fto \psi^\mM(x))$ and \mbox{$\varphi^\mMp(x') = (a \fto \psi^\mMp(x'))$}. By the induction assumption of~\eqref{eq: GDHAW 1}, 
		\( Z(x,x') \leq (\psi^\mM(x) \fto \psi^\mMp(x')). \)
		The assertion~\eqref{eq: GDHAW 1} follows from this and~\eqref{eq: JHFJWx 60}.  

		\item Case $\varphi = \tuple{\alpha}\psi$: Since $\mM$ is witnessed w.r.t.\ \fPDL, there exists $y \in \Delta^\mM$ such that 
		\begin{equation}
		\varphi^\mM(x) = \alpha^\mM(x,y) \fand \psi^\mM(y). \label{eq: HJKAQ 1}
		\end{equation}
		By the induction assumption of~\revisedA{\eqref{eq: GDHAW 2}}, there exists $y' \in \Delta^\mMp$ such that 
		\begin{equation}
		Z(x,x') \fand \alpha^\mM(x,y) \leq \alpha^\mMp(x',y') \fand Z(y,y'). \label{eq: HJKAQ 2}
		\end{equation}
		By definition, 
		\begin{equation}
		\alpha^\mMp(x',y') \fand \psi^\mMp(y') \leq \varphi^\mMp(x'). \label{eq: HJKAQ 2a}
		\end{equation}
		
		By the induction assumption of~\eqref{eq: GDHAW 1}, 
		\begin{equation}
		Z(y,y') \ \leq\ \left(\psi^\mM(y) \fto \psi^\mMp(y')\right). \label{eq: HJKAQ 3}
		\end{equation}
		
		By \eqref{eq: HJKAQ 2} and \eqref{fop: GDJSK 00},  
		\[ Z(x,x') \ \leq\ \left(\alpha^\mM(x,y) \fto Z(y,y') \fand \alpha^\mMp(x',y')\right). \]
		Since $\fand$ is monotone and $\fto$ is monotone w.r.t.\ the second argument, by \eqref{eq: HJKAQ 3}, it follows that 
		\[ Z(x,x') \ \leq\ \left(\alpha^\mM(x,y) \fto \left(\psi^\mM(y) \fto \psi^\mMp(y')\right) \fand \alpha^\mMp(x',y')\right). \]
		Since $\fand$ is commutative and $\fto$ is monotone w.r.t.\ the second argument, by~\eqref{eq: JHFJWx 30}, it follows that 
		\[ Z(x,x') \ \leq\ \left(\alpha^\mM(x,y) \fto \left(\psi^\mM(y) \fto \alpha^\mMp(x',y') \fand \psi^\mMp(y')\right)\right). \]
		By \eqref{eq: JHFJWx 10}, it follows that 
		\[ Z(x,x') \ \leq\  \left(\alpha^\mM(x,y) \fand \psi^\mM(y) \fto \alpha^\mMp(x',y') \fand \psi^\mMp(y')\right). \]
		Since $\fto$ is monotone w.r.t.\ the second argument, by \eqref{eq: HJKAQ 1} and \eqref{eq: HJKAQ 2a}, it follows that 
		\[ Z(x,x') \ \leq\ \left(\varphi^\mM(x) \fto \varphi^\mMp(x')\right). \]

		\item Case $\varphi = [\alpha]\psi$: Since $\mMp$ is witnessed w.r.t.\ \fPDL, there exists $y' \in \Delta^\mMp$ such that 
		\begin{equation}
		\varphi^\mMp(x') = \left(\alpha^\mMp(x',y') \fto \psi^\mMp(y')\right). \label{eq: KDNSJ 1}
		\end{equation}
		By the induction assumption of~\revisedA{\eqref{eq: GDHAW 3}}, there exists $y \in \Delta^\mM$ such that 
		\begin{equation}
		Z(x,x') \fand \alpha^\mMp(x',y') \leq \alpha^\mM(x,y) \fand Z(y,y'). \label{eq: KDNSJ 2}
		\end{equation}
		By definition, 
		\begin{equation}
		\varphi^\mM(x) \ \leq\ \left(\alpha^\mM(x,y) \fto \psi^\mM(y)\right). \label{eq: KDNSJ 2a}
		\end{equation}
		By the induction assumption of~\eqref{eq: GDHAW 1}, 
		\begin{equation}
		Z(y,y') \ \leq\ \left(\psi^\mM(y) \fto \psi^\mMp(y')\right). \label{eq: KDNSJ 3}
		\end{equation}
		
		By \eqref{eq: KDNSJ 2} and \eqref{fop: GDJSK 00}, 
		\[ Z(x,x') \ \leq\ \left(\alpha^\mMp(x',y') \fto Z(y,y') \fand \alpha^\mM(x,y)\right). \]
		Since $\fand$ is monotone and $\fto$ is monotone w.r.t.\ the second argument, by \eqref{eq: KDNSJ 3}, it follows that 
		\[ Z(x,x') \ \leq\ \left(\alpha^\mMp(x',y') \fto \left(\psi^\mM(y) \fto \psi^\mMp(y')\right) \fand \alpha^\mM(x,y)\right). \]
		Since $\fand$ is commutative and $\fto$ is monotone w.r.t.\ the second argument, by \eqref{eq: JHFJWx 20}, it follows that 
		\[ Z(x,x') \ \leq\ \left(\alpha^\mMp(x',y') \fto \left(\left(\alpha^\mM(x,y) \fto \psi^\mM(y)\right) \fto \psi^\mMp(y')\right)\right). \]
		By \eqref{eq: JHFJWx 1}, it follows that 
		\[ Z(x,x') \ \leq\ \left(\left(\alpha^\mM(x,y) \fto \psi^\mM(y)\right) \fto \left(\alpha^\mMp(x',y') \fto \psi^\mMp(y')\right)\right). \]
		Since $\fto$ is antitone w.r.t.\ the first argument, by \eqref{eq: KDNSJ 2a} and \eqref{eq: KDNSJ 1}, it follows that 
		\[ Z(x,x') \ \leq\ \left(\varphi^\mM(x) \fto \varphi^\mMp(x')\right). \]
	\end{itemize}

	Now, consider the assertion~\eqref{eq: GDHAW 2}. Let $x,y \in \Delta^\mM$ and $x' \in \Delta^\mMp$. It suffices to show that there exists $y' \in \Delta^\mMp$ such that
	\begin{equation}\label{eq: HDHAK}
	Z(x,x') \fand \alpha^\mM(x,y) \leq \alpha^\mMp(x',y') \fand Z(y,y').
	\end{equation}
	The base case occurs when $\alpha$ is an atomic program and follows from~\eqref{eq: FDS2}. The induction steps are given below.
	\begin{itemize}
		\item Case $\alpha = \beta \circ \gamma$: Since $\mM$ is witnessed w.r.t.\ \fPDL, there exists $z \in \Delta^\mM$ such that \mbox{$\alpha^\mM(x,y) = \beta^\mM(x,z) \fand \gamma^\mM(z,y)$}. By the induction assumption of~\eqref{eq: GDHAW 2}, there exist $z'$ and $y'$ such that:
		\begin{eqnarray*}
			Z(x,x') \fand \beta^\mM(x,z) & \leq & \beta^\mMp(x',z') \fand Z(z,z') \\
			Z(z,z') \fand \gamma^\mM(z,y) & \leq & \gamma^\mMp(z',y') \fand Z(y,y').
		\end{eqnarray*}
		Since $\fand$ is associative and monotone, it follows that 
		\begin{eqnarray*}
			Z(x,x') \fand \alpha^\mM(x,y) 
			& = & Z(x,x') \fand \beta^\mM(x,z) \fand \gamma^\mM(z,y) \\
			& \leq & \beta^\mMp(x',z') \fand Z(z,z') \fand \gamma^\mM(z,y) \\
			& \leq & \beta^\mMp(x',z') \fand \gamma^\mMp(z',y') \fand Z(y,y') \\
			& \leq & \alpha^\mMp(x',y') \fand Z(y,y').
		\end{eqnarray*}
		
		\item Case $\alpha = \beta \cup \gamma$: Without loss of generality, suppose $\beta^\mM(x,y) \geq \gamma^\mM(x,y)$. Thus, $\alpha^\mM(x,y) = \beta^\mM(x,y)$. By the induction assumption of~\eqref{eq: GDHAW 2}, there exists $y' \in \Delta^\mMp$ such that 
		\[ Z(x,x') \fand \beta^\mM(x,y) \leq \beta^\mMp(x',y') \fand Z(y,y'). \] 
		Thus, 
		\begin{eqnarray*}
		Z(x,x') \fand \alpha^\mM(x,y) & = & Z(x,x') \fand \beta^\mM(x,y)\\
		& \leq & \beta^\mMp(x',y') \fand Z(y,y')\\
		& \leq & \alpha^\mMp(x',y') \fand Z(y,y').
		\end{eqnarray*}

		\item Case $\alpha = \beta^*$: Since $\mM$ is witnessed w.r.t.\ \fPDL, there exist $x_0, \ldots, x_k \in \Delta^\mM$ such that $x_0 = x$, $x_k = y$, and 
		\[ \alpha^\mM(x,y) = \beta^\mM(x_0,x_1) \fand\cdots\fand \beta^\mM(x_{k-1},x_k). \]
		Let $x'_0 = x'$. By the induction assumption of~\eqref{eq: GDHAW 2}, there exist $x'_1,\ldots,x'_k \in \Delta^\mMp$ such that 
		\[ Z(x_i,x'_i) \fand \beta^\mM(x_i,x_{i+1}) \leq \beta^\mMp(x'_i,x'_{i+1}) \fand Z(x_{i+1},x'_{i+1}) \]
		for all $0 \leq i < k$. Since $\fand$ is associative and monotone, it follows that 
		\begin{eqnarray*}
			\!\!\!\!\!& \!\! & Z(x_0,x'_0) \fand \alpha^\mM(x_0,x_k) \\
			\!\!\!\!\!& \!=\! & Z(x_0,x'_0) \fand \beta^\mM(x_0,x_1) \fand\cdots\fand \beta^\mM(x_{k-1},x_k) \\ 
			\!\!\!\!\!& \!\leq\! & \beta^\mMp(x'_0,x'_1) \fand Z(x_1,x'_1) \fand \beta^\mM(x_1,x_2) \fand \cdots\fand \beta^\mM(x_{k-1},x_k) \\
			\!\!\!\!\!& \!\leq\! & \beta^\mMp(x'_0,x'_1) \fand \beta^\mMp(x'_1,x'_2) \fand Z(x_2,x'_2) \fand \beta^\mM(x_2,x_3) \fand \cdots\fand \beta^\mM(x_{k-1},x_k) \\
			\!\!\!\!\!& \!\leq\! & \ldots \\
			\!\!\!\!\!& \!\leq\! & \beta^\mMp(x'_0,x'_1) \fand\cdots\fand \beta^\mMp(x'_{k-1},x'_k) \fand Z(x_k,x'_k) \\ 
			\!\!\!\!\!& \!\leq\! & \alpha^\mMp(x'_0,x'_k) \fand Z(x_k,x'_k).
		\end{eqnarray*}
		Taking $y' = x'_k$, we obtain~\eqref{eq: HDHAK}.
	\end{itemize}
	
	The assertion~\eqref{eq: GDHAW 3} can be proved analogously to~\eqref{eq: GDHAW 2}.
\myend
\end{proof}
} 

The following lemma is a counterpart of Lemma~\ref{lemma: GDHAW} for \fK. It is needed for the proof of Theorem~\ref{theorem: HM}. Its proof proceeds analogously to that of assertion~\eqref{eq: GDHAW 1} in Lemma~\ref{lemma: GDHAW}, using~\eqref{eq: FDS2} and~\eqref{eq: FDS3} in place of~\eqref{eq: GDHAW 2} and~\eqref{eq: GDHAW 3}, respectively.

\begin{lemma} \label{lemma: GDHAW 2}
Let $\mM$ and $\mM'$ be fuzzy Kripke models that are witnessed w.r.t.\ \fK and let $Z$ be a fuzzy directed simulation between $\mM$ and $ \mM'$. Then, the following property holds for every $x \in \Delta^\mM$, $x' \in \Delta^\mMp$, and every positive formula $\varphi$ of \fK:
\[ Z(x,x') \ \leq\ \left(\varphi^\mM(x) \fto \varphi^\mMp(x')\right). \]
\end{lemma}


Here is the Hennessy--Milner property of fuzzy directed simulations:

\begin{theorem} \label{theorem: HM}
Suppose $\fand$ is right-continuous and let $\mM$ and $\mMp$ be image-finite fuzzy Kripke models. Then the fuzzy relation $Z : \Delta^\mM \times \Delta^\mMp \to L$ defined by 
\[ Z(x,x') = \bigwedge \{\varphi^\mM(x) \fto \varphi^\mMp(x') \mid \textrm{$\varphi$ is a positive formula of \fK}\} \] 
is the greatest fuzzy directed simulation between $\mM$ and~$\mMp$.	
\end{theorem}

\newcommand{\ProofTheoremHM}{
\begin{proof}[of Theorem~\ref{theorem: HM}]
By Lemma~\ref{lemma: GDHAW 2}, it suffices to prove that $Z$ is a fuzzy directed simulation between $\mM$ and $\mMp$. 
By definition, $Z$ satisfies the condition~\eqref{eq: FDS1}.

We prove that $Z$ satisfies the condition \eqref{eq: FDS2} by adapting the corresponding part of the proof of~\cite[Theorem 5.2]{FBSML} (on the Hennessy--Milner property of fuzzy bisimulations). Let $\varrho \in \SA$, $x,y \in \Delta^\mM$ and $x' \in \Delta^\mMp$. Let $a = Z(x,x') \fand \varrho^\mM(x,y)$. For a contradiction, suppose that, for every $y' \in \Delta^\mMp$, \mbox{$a > \varrho^\mMp(x',y') \fand Z(y,y')$}. Since $\fand$ is right-continuous, by the definition of $Z(y,y')$, it follows that, for every $y' \in \Delta^\mMp$, there exists a positive formula $\varphi_{y'}$ of \fK such that
\[ a > \varrho^\mMp(x',y') \fand \left(\varphi_{y'}^\mM(y) \fto \varphi_{y'}^\mMp(y')\right). \]
For every $y' \in \Delta^\mMp$, let $\psi_{y'} = (\varphi_{y'}^\mM(y) \to \varphi_{y'})$. 
Let $\Phi = \{\psi_{y'} \mid y' \in \Delta^\mMp, \varrho^\mMp(x',y') > 0\}$. It is finite, since $\mMp$ is image-finite. 
Observe that, for every $y' \in \Delta^\mMp$, $\psi_{y'}^\mM(y) = 1$ and 
\begin{equation}\label{eq: HSKZK}
a > \varrho^\mMp(x',y') \fand \psi_{y'}^\mMp(y'). 
\end{equation}
Let $\varphi = \tuple{\varrho}\bigwedge\!\Phi$. It is a positive formula of \fK. 
Thus, $\varphi^\mM(x) \geq \varrho^\mM(x,y)$ since $(\bigwedge\!\Phi)^\mM(y) = 1$. Since $\fand$ is monotone and $\mMp$ is image-finite, by~\eqref{eq: HSKZK}, we have $a > \varphi^\mMp(x')$, which means 
\[ Z(x,x') \fand \varrho^\mM(x,y) > \varphi^\mMp(x'). \]
Since $\varphi^\mM(x) \geq \varrho^\mM(x,y)$ and $\fand$ is monotone, it follows that  
\[ Z(x,x') \fand \varphi^\mM(x) > \varphi^\mMp(x'). \]
By~\eqref{fop: GDJSK 00}, this implies that 
\[ Z(x,x') > (\varphi^\mM(x) \fto \varphi^\mMp(x')), \]
which contradicts the definition of $Z(x,x')$. 

We now prove that $Z$ satisfies the condition \eqref{eq: FDS3}. 
Let $\varrho \in \SA$, $x \in \DeltaM$, and $x',y' \in \DeltaMp$. 
For a contradiction, suppose that, for every $y \in \DeltaM$, 
\[ Z(x,x') \fand \varrho^\mMp(x',y') > \varrho^\mM(x,y) \fand Z(y,y'). \]
Since $\fand$ is right-continuous, by the definition of $Z$, for every $y \in \DeltaM$, there exists a positive formula $\varphi_y$ of \fK such that 
\begin{equation}\label{eq: JHEKA}
Z(x,x') \fand \varrho^\mMp(x',y') > \varrho^\mM(x,y) \fand (\varphi_y^\mM(y) \fto \varphi_y^\mMp(y')).
\end{equation}
Consider an arbitrary $y \in \DeltaM$. 
Let $\psi_y = \varrho^\mM(x,y) \fand (\varphi_y^\mM(y) \to \varphi_y)$. It is a positive formula of \fK. 
We have $\psi_y^\mM(y) = \varrho^\mM(x,y)$ and, by~\eqref{eq: JHEKA}, 
\begin{equation}\label{eq: KJRMZ}
\psi_y^\mMp(y') < Z(x,x') \fand \varrho^\mMp(x',y').
\end{equation}
Let $\Psi = \{\psi_y \mid y \in \DeltaM, \varrho^\mM(x,y) > 0\}$. It is finite, since $\mM$ is image-finite. By~\eqref{eq: KJRMZ}, we have 
\begin{equation}\label{eq: JHMSP}
(\textstyle\bigvee\!\Psi)^\mMp(y') < Z(x,x') \fand \varrho^\mMp(x',y').
\end{equation}
For every $y \in \DeltaM$ such that $\varrho^\mM(x,y) > 0$, we have $(\textstyle\bigvee\!\Psi)^\mM(y) \geq \psi_y^\mM(y) = \varrho^\mM(x,y)$. Hence
\begin{equation}\label{eq: HEHAK}
([\varrho]\textstyle\bigvee\!\Psi)^\mM(x) = \displaystyle\bigwedge_{y \in \DeltaM}\!\! \big(\varrho^\mM(x,y) \fto (\textstyle\bigvee\!\Psi)^\mM(y)\big) = 1.
\end{equation}
By the definition of $Z$ together with~\eqref{eq: HEHAK} and~\eqref{eq: JHFJWx 0}, we have
\begin{equation}\label{eq: JHJHK}
Z(x,x') \leq \big(([\varrho]\textstyle\bigvee\!\Psi)^\mM(x) \fto ([\varrho]\textstyle\bigvee\!\Psi)^\mMp(x')\big) = ([\varrho]\textstyle\bigvee\!\Psi)^\mMp(x'). \\
\end{equation}
Since $\fand$ is monotone, it follows that 
\begin{eqnarray*}
Z(x,x') \fand \varrho^\mMp(x',y') & \leq & ([\varrho]\textstyle\bigvee\!\Psi)^\mMp(x') \fand \varrho^\mMp(x',y') \\
& \leq & \varrho^\mMp(x',y') \fand \displaystyle\bigwedge_{y'' \in \DeltaMp}\!\!\!\! \big(\varrho^\mMp(x',y'') \fto (\textstyle\bigvee\!\Psi)^\mMp(y'')\big) \\
& \leq & \varrho^\mMp(x',y') \fand \big(\varrho^\mMp(x',y') \fto (\textstyle\bigvee\!\Psi)^\mMp(y')\big) \\
& \leq & (\textstyle\bigvee\!\Psi)^\mMp(y'),
\end{eqnarray*}
which contradicts~\eqref{eq: JHMSP} and completes the proof.
\myend
\end{proof}
} 

The following corollary follows directly from Theorem~\ref{theorem: HM} and Lemma~\ref{lemma: GDHAW}. 

\begin{corollary} \label{cor: HM}
Suppose $\fand$ is right-continuous. 
Let $\mM$ and $\mMp$ be image-finite fuzzy Kripke models that are witnessed w.r.t.\ \fPDL. 
Then, for every $x \in \Delta^\mM$ and $x' \in \Delta^\mMp$, 
\begin{eqnarray*}
\!\!\!\!\!&&  
\bigwedge\{\varphi^\mM(x) \fto \varphi^\mMp(x') \mid \textrm{$\varphi$ is a positive formula of \fK}\} \\ 
\!\!\!\!\!&\!=\!& 
\bigwedge\{\varphi^\mM(x) \fto \varphi^\mMp(x') \mid \textrm{$\varphi$ is a positive formula of \fPDL}\}. 
\end{eqnarray*}
\end{corollary}

{\markRevisedA
We give below some remarks:
\begin{itemize}
\item Our proof of Theorem~\ref{theorem: HM} given in the appendix requires that the considered logic has the formula constructor \mbox{$a \fand \varphi$}. We do not know whether this theorem can be strengthened to a form in which \fK is replaced with \fKx. That is, we do not have a proof of Theorem~\ref{theorem: HM} without using this formula constructor.

\item Allowing the formula constructor \mbox{$a \fand \varphi$}, Theorem~\ref{theorem: preservation} is stronger than its version in which \fPDL is replaced with \fPDLx.

\item Theorem~\ref{theorem: preservation} cannot be strengthened by replacing \fPDL with its extension in which the formula constructor \mbox{$a \fand \varphi$} is replaced with the more general form \mbox{$\varphi \fand \psi$} or \mbox{$\varphi \,\&\, \psi$}, whose semantics is specified by
\[ (\varphi \fand \psi)^\mM(x) = (\varphi \,\&\, \psi)^\mM(x) = \varphi^\mM(x) \fand \psi^\mM(x), \]
for a given fuzzy Kripke model $\mM$ and $x \in \DeltaM$. A counterexample similar to that in \cite[Remark~4.6]{FBSML} is as follows. 
Let $L = [0,1]$, $\SA = \emptyset$, $\SP = {p}$, and $\varphi = p \fand p$. Let $\mM$ and $\mMp$ be fuzzy Kripke models such that $\DeltaM = {v}$, $p^\mM(v) = 1$, $\DeltaMp = {v'}$, and $p^\mMp(v') = 0.5$. We have $\varphi^\mM(v) = 1$. The greatest fuzzy directed simulation $Z$ between $\mM$ and $\mMp$ satisfies
\[ Z(v,v') = (p^\mM(v) \fto p^\mMp(v')) = (1 \fto 0.5) = 0.5. \]
If $\mL$ is the product structure, then
\[ (\varphi^\mM(v) \fto \varphi^\mMp(v')) = (1 \fto 0.25) = 0.25. \]
If $\mL$ is the {\L}ukasiewicz structure, then
\[ (\varphi^\mM(v) \fto \varphi^\mMp(v')) = (1 \fto 0) = 0. \]
Thus, in both cases,
\[ Z(v,v') \not\leq (\varphi^\mM(v) \fto \varphi^\mMp(v')). \]

\item As a counterpart of Theorem~\ref{theorem: preservation}, \cite[Theorem~4.1]{FBSML} concerns the invariance of formulas under fuzzy bisimulations. It is formulated for fuzzy modal logics without the formula constructor \mbox{$a \fand \varphi$}, but can be strengthened for fuzzy modal logics with this formula constructor. For such an extension, the proof of the theorem can be updated appropriately, using~\eqref{eq: JHFJWx 66} analogously to the use of~\eqref{eq: JHFJWx 65} in the proof of Theorem~\ref{theorem: preservation}.
\end{itemize}
}


\section{Computing fuzzy directed simulations}
\label{section: computation}

In this section, we show that the problem of computing the greatest fuzzy directed simulation between two finite fuzzy Kripke models can be reduced to the problem of computing the greatest correct marking of a finite fuzzy minimax net.

Let $\mM$ and $\mM'$ be fuzzy Kripke models \revisedA{with disjoint domains}. 
The {\em minimax net corresponding to $\tuple{\mM,\mM'}$ via directed simulation} is $\mN = \tuple{\Vmin, \Vmax, E, \Mz}$ specified as follows:
	\begin{itemize}
		\item $\Vmin = \DeltaM \times \DeltaMp$, 
		\item $\Vmax = (\DeltaM \times \DeltaMp \times \SA) \cup (\DeltaMp \times \DeltaM \times \SA)$, 
		\item $E: (\Vmin \times \Vmax) \cup (\Vmax \times \Vmin) \to L$ is the following fuzzy set
		\[
		\begin{array}{l}
		\{\tuple{\tuple{y,y'},\tuple{x',y,\varrho}} : \varrho^\mMp(x',y') \mid y \in \DeltaM, x',y' \in \DeltaMp, \varrho \in \SA\}\ \lor \\
		\{\tuple{\tuple{x',y,\varrho},\tuple{x,x'}} : \varrho^\mM(x,y) \mid x,y \in \DeltaM, x' \in \DeltaMp, \varrho \in \SA\}\ \lor \\
		\{\tuple{\tuple{y,y'},\tuple{x,y',\varrho}} : \varrho^\mM(x,y) \mid x,y \in \DeltaM, y' \in \DeltaMp, \varrho \in \SA\}\ \lor \\ 
		\{\tuple{\tuple{x,y',\varrho},\tuple{x,x'}} : \varrho^\mMp(x',y') \mid x \in \DeltaM, x',y' \in \DeltaMp, \varrho \in \SA\}, 
		\end{array}
		\]
		\item $\Mz: \Vmin \to L$ is the fuzzy set such that, for every $x \in \DeltaM$ and $x' \in \DeltaMp$:
		\begin{equation}\label{eq: JHDLA}
		\Mz(\tuple{x,x'}) = \bigwedge_{p \in \SP}\!\! \left(p^\mM(x) \fto p^\mMp(x')\right).
		\end{equation} 
	\end{itemize}
    
The notion defined above differs from the notion of a minimax net bisimulatedly corresponding to a pair of fuzzy labeled graphs \cite{DBLP:journals/tfs/NguyenMS23}, among others, in that $\fto$ is used in~\eqref{eq: JHDLA} instead of~$\fequiv$.

{\markRevisedA
\begin{remark}\label{remark: JHZHA}
The assumption that $\mM$ and $\mMp$ have disjoint domains is needed for the above definition to ensure that 
\begin{itemize}
\item $E(\tuple{y,y'},\tuple{x',y,\varrho}) = \varrho^\mMp(x',y')$, for $y \in \DeltaM$, $x',y' \in \DeltaMp$, and $\varrho \in \SA$;
\item $E(\tuple{x',y,\varrho},\tuple{x,x'}) = \varrho^\mM(x,y)$, for $x,y \in \DeltaM$, $x' \in \DeltaMp$, and $\varrho \in \SA$;
\item $E(\tuple{y,y'},\tuple{x,y',\varrho}) = \varrho^\mM(x,y)$, for $x,y \in \DeltaM$, $y' \in \DeltaMp$, and $\varrho \in \SA$;
\item $E(\tuple{x,y',\varrho},\tuple{x,x'}) = \varrho^\mMp(x',y')$, for $x \in \DeltaM$, $x',y' \in \DeltaMp$, and $\varrho \in \SA$.
\end{itemize}
Another approach is to discard this assumption and redefine the notion of the minimax net corresponding to $\tuple{\mM,\mM'}$ via directed simulation, with the following modifications: 
	\begin{itemize}
		\item $\Vmax = (\DeltaM \times \DeltaMp \times \SA \times \{0\}) \cup (\DeltaMp \times \DeltaM \times \SA \times \{1\})$, 
		\item $E: (\Vmin \times \Vmax) \cup (\Vmax \times \Vmin) \to L$ is the following fuzzy set
		\[
		\begin{array}{l}
		\{\tuple{\tuple{y,y'},\tuple{x',y,\varrho,1}} : \varrho^\mMp(x',y') \mid y \in \DeltaM, x',y' \in \DeltaMp, \varrho \in \SA\}\ \lor \\
		\{\tuple{\tuple{x',y,\varrho,1},\tuple{x,x'}} : \varrho^\mM(x,y) \mid x,y \in \DeltaM, x' \in \DeltaMp, \varrho \in \SA\}\ \lor \\
		\{\tuple{\tuple{y,y'},\tuple{x,y',\varrho,0}} : \varrho^\mM(x,y) \mid x,y \in \DeltaM, y' \in \DeltaMp, \varrho \in \SA\}\ \lor \\ 
		\{\tuple{\tuple{x,y',\varrho,0},\tuple{x,x'}} : \varrho^\mMp(x',y') \mid x \in \DeltaM, x',y' \in \DeltaMp, \varrho \in \SA\}. 
		\end{array}
		\]
	\end{itemize}
Such modifications would complicate the subsequent exposition, and for this reason, we prefer to adopt the aforementioned assumption.
\myend
\end{remark}
}

\begin{figure}[t!]
\begin{center}
\begin{tabular}{|c|}
\hline
\begin{tikzpicture}[->,>=stealth,auto,black]
\node (sta) {};
\node (uvp) [draw, node distance=1.2cm, below of=sta] {$\tuple{u,v'}\!:\!0$}; 
    \node (xuvp) [node distance=10cm, right of=uvp] {};
    \node (uvpr) [node distance=0.7cm, above of=xuvp] {$\tuple{u,v',\varrho}$};
    \node (vpur) [node distance=0.7cm, below of=xuvp] {$\tuple{v',u,\varrho}$};
\node (vvp) [draw, node distance=2.8cm, below of=uvp] {$\tuple{v,v'}\!:\!{\scriptstyle(0.7 \fto 0.5)}$}; 
    \node (xvvp) [node distance=10cm, right of=vvp] {};
    \node (vvpr) [node distance=0.7cm, above of=xvvp] {$\tuple{v,v',\varrho}$};
    \node (vpvr) [node distance=0.7cm, below of=xvvp] {$\tuple{v',v,\varrho}$};
\node (wvp) [draw, node distance=2.8cm, below of=vvp] {$\tuple{w,v'}\!:\!1$}; 
    \node (xwvp) [node distance=10cm, right of=wvp] {};
    \node (wvpr) [node distance=0.7cm, above of=xwvp] {$\tuple{w,v',\varrho}$};
    \node (vpwr) [node distance=0.7cm, below of=xwvp] {$\tuple{v',w,\varrho}$};
\node (uup) [draw, node distance=2.8cm, below of=wvp] {$\tuple{u,u'}\!:\!1$}; 
    \node (xuup) [node distance=10cm, right of=uup] {};
    \node (uupr) [node distance=0.7cm, above of=xuup] {$\tuple{u,u',\varrho}$};
    \node (upur) [node distance=0.7cm, below of=xuup] {$\tuple{u',u,\varrho}$};
\node (vup) [draw, node distance=2.8cm, below of=uup] {$\tuple{v,u'}\!:\!0$}; 
    \node (xvup) [node distance=10cm, right of=vup] {};
    \node (vupr) [node distance=0.7cm, above of=xvup] {$\tuple{v,u',\varrho}$};
    \node (upvr) [node distance=0.7cm, below of=xvup] {$\tuple{u',v,\varrho}$};
\node (wup) [draw, node distance=2.8cm, below of=vup] {$\tuple{w,u'}\!:\!0$}; 
    \node (xwup) [node distance=10cm, right of=wup] {};
    \node (wupr) [node distance=0.7cm, above of=xwup] {$\tuple{w,u',\varrho}$};
    \node (upwr) [node distance=0.7cm, below of=xwup] {$\tuple{u',w,\varrho}$};
\node (stp) [node distance=1.0cm, below of=wup] {};
\draw (uvp) to node [above]{\footnotesize{0.7}} (vpur);	
\draw (uvp) to node [below, pos=0.09, yshift=-3]{\footnotesize{0.8}} (upur);	
\draw (vvp) to node [above, pos=0.80]{\footnotesize{0.7}} (uvpr);	
\draw (vvp) to node [below, pos=0.78, yshift=1]{\footnotesize{0.7}} (vpvr);	
\draw (vvp) to node [above, pos=0.92, yshift=-1]{\footnotesize{0.8}} (wvpr);	
\draw (vvp) to node [above, pos=0.57, yshift=3]{\footnotesize{0.8}} (upvr);	
\draw (wvp) to node [above, pos=0.67]{\footnotesize{0.9}} (uvpr);	
\draw (wvp) to node [above, pos=0.25, yshift=-1]{\footnotesize{0.6}} (vvpr);	
\draw (wvp) to node [below, pos=0.45, yshift=1]{\footnotesize{0.7}} (vpwr);	
\draw (wvp) to node [below, pos=0.04, yshift=-3]{\footnotesize{0.8}} (upwr);	
\draw (vup) to node [above, pos=0.81]{\footnotesize{0.7}} (uupr);	
\draw (vup) to node [above, pos=0.12, yshift=-1]{\footnotesize{0.8}} (wupr);	
\draw (wup) to node [below, pos=0.87, yshift=-1]{\footnotesize{0.9}} (uupr);	
\draw (wup) to node [below, pos=0.80]{\footnotesize{0.6}} (vupr);	
\draw[dashed] (uvpr) to node [above]{\footnotesize{0.7}} (uvp);	
\draw[dashed] (uvpr) to node [below, pos=0.78, yshift=-3]{\footnotesize{0.8}} (uup);	
\draw[dashed] (vpvr) to node [above, pos=0.80]{\footnotesize{0.7}} (uvp);	
\draw[dashed] (vvpr) to node [above, pos=0.12, yshift=-1]{\footnotesize{0.7}} (vvp);
\draw[dashed] (vpvr) to node [below, pos=0.37, yshift=1]{\footnotesize{0.8}} (wvp);	
\draw[dashed] (vvpr) to node [above, pos=0.65, yshift=3]{\footnotesize{0.8}} (vup);
\draw[dashed] (vpwr) to node [above, pos=0.80]{\footnotesize{0.9}} (uvp);	
\draw[dashed] (vpwr) to node [below, pos=0.88]{\footnotesize{0.6}} (vvp);	
\draw[dashed] (wvpr) to node [below, pos=0.55, yshift=1]{\footnotesize{0.7}} (wvp);	
\draw[dashed] (wvpr) to node [above, pos=0.93, yshift=3]{\footnotesize{0.8}} (wup);	
\draw[dashed] (upvr) to node [above, pos=0.88]{\footnotesize{0.7}} (uup);	
\draw[dashed] (upvr) to node [below, pos=0.75]{\footnotesize{0.8}} (wup);	
\draw[dashed] (upwr) to node [below, pos=0.94, yshift=-1]{\footnotesize{0.9}} (uup);	
\draw[dashed] (upwr) to node [below, pos=0.23]{\footnotesize{0.6}} (vup);	
\end{tikzpicture}
\\
\hline
\end{tabular}
\caption{Illustration of the minimax net mentioned in Example~\ref{example: JHIUS}.\label{fig: JTKXH}}
\end{center}
\end{figure}
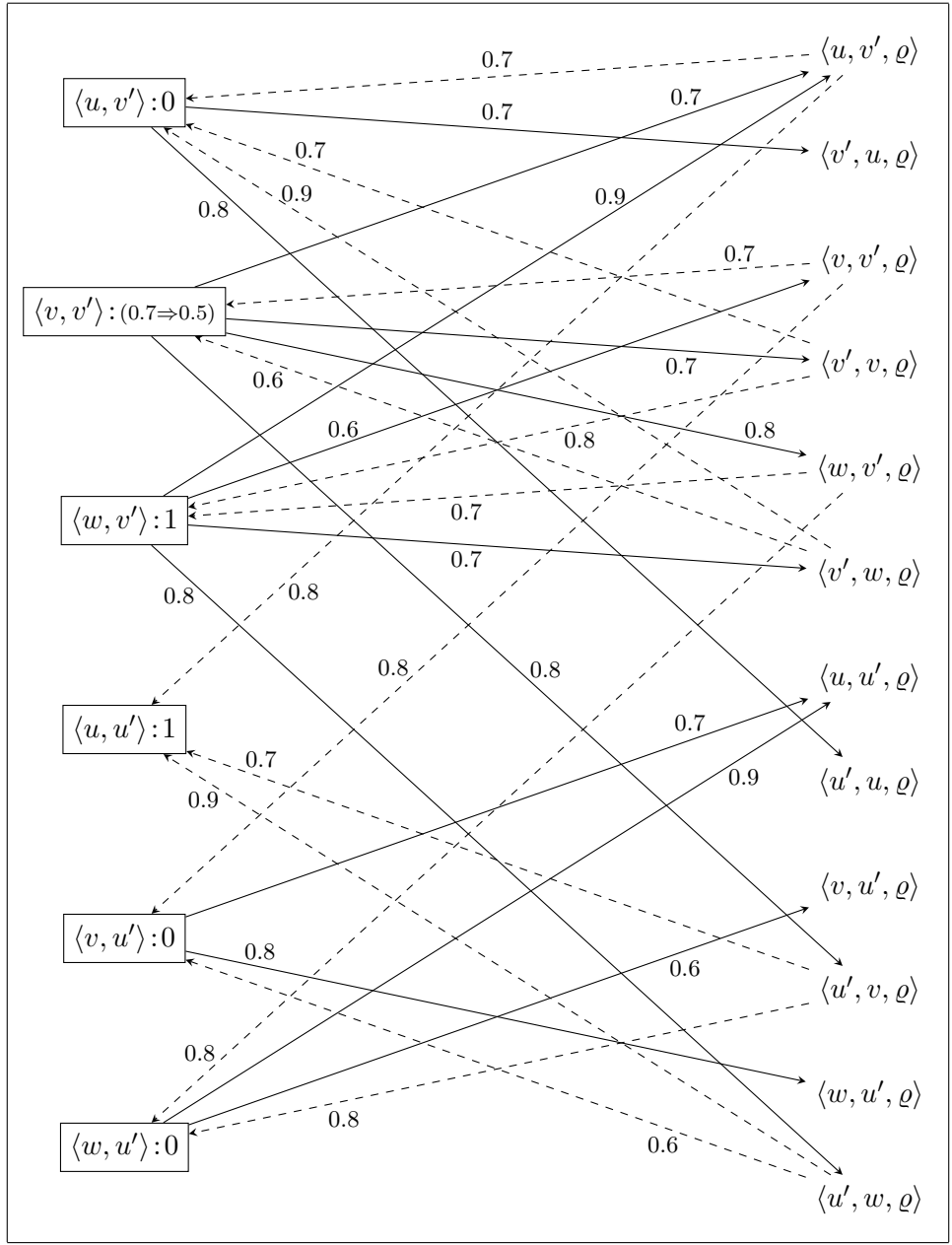

\begin{example}\label{example: JHIUS}
Reconsider the signature $\tuple{\SA,\SP}$ and the fuzzy Kripke models $\mM$ and $\mMp$ specified in Example~\ref{example: JHRJA}, which are illustrated in Figure~\ref{fig: HFJSK}. The minimax net $\mN = \tuple{\Vmin, \Vmax, E, \Mz}$ corresponding to $\tuple{\mM,\mM'}$ via directed simulation is depicted in Figure~\ref{fig: JTKXH}. Its min-nodes and max-nodes are shown on the left-hand side and the right-hand side of the figure, respectively. The label next to each min-node denotes its marking limit (e.g., $\Mz(\tuple{v,v'}) = (0.7 \fto 0.5)$). Only positive edges of $\mN$ are shown: edges in $\Vmin \times \Vmax$ are drawn with solid lines, while edges in $\Vmax \times \Vmin$ are drawn with dashed lines. The label of an edge from a node $x$ to a node $y$ denotes the value $E(x,y)$. Note that the min-nodes $\tuple{u,v'}$, $\tuple{v,u'}$, and $\tuple{w,u'}$ have a marking limit of~0, so any edges entering or leaving them can be ignored, as they do not affect the computation of the greatest correct marking of~$\mN$.
\myend
\end{example}

%

The following lemma relates the notion of a fuzzy directed simulation between two fuzzy Kripke models to the notion of a correct marking of the corresponding minimax net. It is a counterpart of \cite[Lemma~11]{DBLP:journals/tfs/NguyenMS23}, which deals with fuzzy bisimulations rather than fuzzy directed simulations.

\begin{lemma}\label{lemma: HDJHA}
Let $\mM$ and $\mMp$ be image-finite fuzzy Kripke models \revisedA{with disjoint domains}, and let $\mN = \tuple{\Vmin, \Vmax, E, \Mz}$ be the minimax net corresponding to $\tuple{\mM,\mM'}$ via directed simulation. Let $Z$ be a fuzzy subset of $\DeltaM \times \DeltaMp$, and let $M: \Vmin \cup \Vmax \to L$ be the following fuzzy set 
\[
\begin{array}{l}
Z \lor \{\tuple{x',y,\varrho} : (\varrho^\mMp \circ Z^-)(x',y) \mid x' \in \DeltaMp, y \in \DeltaM, \varrho \in \SA\} \\
\quad \lor\, \{\tuple{x,y',\varrho} : (\varrho^\mM \circ Z)(x,y') \mid x \in \DeltaM, y' \in \DeltaMp, \varrho \in \SA\}.
\end{array}
\]
Then $Z$ is a fuzzy directed simulation between $\mM$ and $\mMp$ iff $M$ is a correct marking of~$\mN$. 
\end{lemma}	

The proof of this lemma is provided in the appendix. 

\newcommand{\ProofLemmaHDJHA}{
\begin{proof}[of Lemma~\ref{lemma: HDJHA}]
This proof is adapted from the proof of \cite[Lemma~11]{DBLP:journals/tfs/NguyenMS23} (which concerns fuzzy bisimulations between fuzzy labeled graphs).

For the ``only if'' direction, suppose that $Z$ is a fuzzy directed simulation between $\mM$ and $\mMp$. To prove that $M$ is a correct marking of~$\mN$, we need to show that, for every $x,y \in \DeltaM$, $x',y' \in \DeltaMp$, and $\varrho \in \SA$, 
\begin{eqnarray}
M(\tuple{x,x'}) & \leq & \Mz(\tuple{x,x'}) \label{eq: HFKJX 1} \\
M(\tuple{x,x'}) & \leq & \bigwedge_{y \in \DeltaM, \varrho \in \SA}\!\!\!\!\!\!\! \left(E(\tuple{x',y,\varrho},\tuple{x,x'}) \fto M(\tuple{x',y,\varrho})\right) \label{eq: HFKJX 2a}\\
M(\tuple{x,x'}) & \leq & \!\bigwedge_{y' \in \DeltaMp, \varrho \in \SA}\!\!\!\!\!\!\!\! \left(E(\tuple{x,y',\varrho},\tuple{x,x'}) \fto M(\tuple{x,y',\varrho})\right) \label{eq: HFKJX 2b}\\
M(\tuple{x',y,\varrho}) & \leq & \!\bigvee_{y' \in \DeltaMp}\!\! \left(E(\tuple{y,y'},\tuple{x',y,\varrho}) \fand M(\tuple{y,y'})\right) \label{eq: HFKJX 3} \\ 
M(\tuple{x,y',\varrho}) & \leq & \bigvee_{y \in \DeltaM}\! \left(E(\tuple{y,y'},\tuple{x,y',\varrho}) \fand M(\tuple{y,y'})\right). \label{eq: HFKJX 4} 
\end{eqnarray}
By the definition of $M$ and $\mN$, these inequalities are equivalent to the following ones, respectively:
\begin{eqnarray}
Z(x,x') & \leq & \bigwedge_{p \in \SP}\! \left(p^\mM(x) \fto p^\mMp(x')\right) \label{eq: HDHAK 1} \\
Z(x,x') & \leq & \left(\varrho^\mM(x,y) \fto (\varrho^\mMp \circ Z^-)(x',y)\right) \label{eq: HDHAK 2} \\
Z(x,x') & \leq & \left(\varrho^\mMp(x',y') \fto (\varrho^\mM \circ Z)(x,y')\right) \label{eq: HDHAK 3} \\
(\varrho^\mMp \circ Z^-)(x',y) & \leq & \!\!\bigvee_{y' \in \DeltaMp}\!\! \left(\varrho^\mMp(x',y') \fand Z(y,y')\right) \label{eq: HDHAK 4} \\ 
(\varrho^\mM \circ Z)(x,y') & \leq & \!\bigvee_{y \in \DeltaM}\! \left(\varrho^\mM(x,y) \fand Z(y,y')\right). \label{eq: HDHAK 5} 
\end{eqnarray}
The assertion~\eqref{eq: HDHAK 1} follows from~\eqref{eq: FDS1}. 
The assertion~\eqref{eq: HDHAK 2} follows from~\eqref{fop: GDJSK 00} and~\eqref{eq: FDS2}. 
The assertion~\eqref{eq: HDHAK 3} follows from~\eqref{fop: GDJSK 00} and~\eqref{eq: FDS3}. 
The assertions~\eqref{eq: HDHAK 4} and~\eqref{eq: HDHAK 5} clearly hold.  
	
For the ``if'' direction, suppose that $M$ is a correct marking of~$\mN$, i.e., \eqref{eq: HFKJX 1}--\eqref{eq: HFKJX 4} hold for all $x,y \in \DeltaM$, $x',y' \in \DeltaMp$ and $\varrho \in \SA$. We need to show that $Z$ satisfies the conditions~\eqref{eq: FDS1}--\eqref{eq: FDS3}. The assertions~\eqref{eq: HFKJX 1}--\eqref{eq: HFKJX 4} imply \eqref{eq: HDHAK 1}--\eqref{eq: HDHAK 5}. The condition~\eqref{eq: FDS1} follows from \eqref{eq: HDHAK 1}. The condition~\eqref{eq: FDS2} follows from \eqref{eq: HDHAK 2}, \eqref{fop: GDJSK 00}, and the assumption that $\mMp$ is image-finite. Similarly, the condition~\eqref{eq: FDS3} follows from \eqref{eq: HDHAK 3}, \eqref{fop: GDJSK 00}, and the assumption that $\mM$ is image-finite.
\myend
\end{proof}
} 

The following theorem reduces the problem of computing the greatest fuzzy directed simulation between two finite fuzzy Kripke models to the problem of computing the greatest correct marking of the corresponding fuzzy minimax net. It is a counterpart of \cite[Theorem~12]{DBLP:journals/tfs/NguyenMS23}, which deals with fuzzy bisimulations rather than fuzzy directed simulations.

\begin{theorem}\label{theorem: HDFMX}
Let $\mM$ and $\mMp$ be image-finite fuzzy Kripke models \revisedA{with disjoint domains}, and let $\mN = \tuple{\Vmin, \Vmax, E, \Mz}$ be the minimax net corresponding to $\tuple{\mM,\mM'}$ via directed simulation. Let $M$ be the greatest correct marking of $\mN$. Then $M|_{\Vmin}$ is the greatest fuzzy directed simulation between~$\mM$ and~$\mMp$.  
\end{theorem}

The proof of this theorem relies on Lemma~\ref{lemma: HDJHA} and is presented in the appendix.

\newcommand{\ProofTheoremHDFMX}{
\begin{proof}[of Theorem~\ref{theorem: HDFMX}]
This proof is adapted from the proof of \cite[Theorem~12]{DBLP:journals/tfs/NguyenMS23} (which concerns fuzzy bisimulations between fuzzy labeled graphs).

Let $Z = M|_{\Vmin}$. By~\eqref{eq: JHDSD 5} and the fact that the greatest correct marking of a minimax net is stable~\cite[Proposition~5]{DBLP:journals/tfs/NguyenMS23}, for every $x,y \in \DeltaM$, $x',y' \in \DeltaMp$ and $\varrho \in \SA$, we have 
\begin{eqnarray*}
	M(\tuple{x',y,\varrho}) & = & \!\bigvee_{y' \in \DeltaMp}\! \left(\varrho^\mMp(x',y') \fand Z(y,y')\right), \\ 
	M(\tuple{x,y',\varrho}) & = & \bigvee_{y \in \DeltaM} \left(\varrho^\mM(x,y) \fand Z(y,y')\right),
\end{eqnarray*}
which mean
\begin{eqnarray*}
	M(\tuple{x',y,\varrho}) & = & (\varrho^\mMp \circ Z^-)(x',y), \\
	M(\tuple{x,y',\varrho}) & = & (\varrho^\mM \circ Z)(x,y').
\end{eqnarray*}
Thus, $Z$ and $M$ satisfy the conditions stated in Lemma~\ref{lemma: HDJHA}. Therefore, $Z$ is a fuzzy directed simulation between $\mM$ and $\mMp$. For a contradiction, assume that there exists a fuzzy directed simulation $Z_2$ between $\mM$ and $\mMp$ such that $Z \lneq Z_2$. Let $M_2$ be defined as $M$ in Lemma~\ref{lemma: HDJHA} but using $Z_2$ instead of~$Z$. By Lemma~\ref{lemma: HDJHA}, $M_2$ is a correct marking of $\mN$, and we have $M \lneq M_2$, which contradicts the assumption that $M$ is the greatest correct marking of~$\mN$. Therefore, $Z$ is the greatest fuzzy directed simulation between~$\mM$ and~$\mMp$.
\myend
\end{proof}
} 

\begin{example}\label{example: JHRJS}
Suppose that the underlying residuated lattice is the G\"odel structure. 
Reconsider the fuzzy Kripke models $\mM$ and $\mMp$ specified in Example~\ref{example: JHRJA} and illustrated in Figure~\ref{fig: HFJSK}, and the minimax net $\mN = \tuple{\Vmin, \Vmax, E, \Mz}$ corresponding to $\tuple{\mM,\mM'}$ via directed simulation, illustrated in Figure~\ref{fig: JTKXH} and discussed in Example~\ref{example: JHIUS}. We have 
\[ \Mz(\tuple{v,v'}) = (0.7 \fto 0.5) = 0.5. \] 
Suppose that $M$ is the greatest correct marking of $\mN$. By~\eqref{eq: JHDSD 3}, we obtain
\begin{eqnarray*}
&& M(\tuple{v',v,\varrho}) \leq E(\tuple{v,v'}, \tuple{v',v,\varrho}) \fand M(\tuple{v,v'}) = 0.7 \fand 0.5 = 0.5, \\
&& M(\tuple{u',v,\varrho}) \leq E(\tuple{v,v'}, \tuple{u',v,\varrho}) \fand M(\tuple{v,v'}) = 0.8 \fand 0.5 = 0.5.
\end{eqnarray*}
Consequently, by~\eqref{eq: JHDSD 2}, we have
\begin{eqnarray*}
&& M(\tuple{w,v'}) \leq E(\tuple{v',v,\varrho}, \tuple{w,v'}) \fto M(\tuple{v',v,\varrho}) \leq (0.8 \fto 0.5) = 0.5, \label{eq: JUZJA 1} \\
&& M(\tuple{u,u'}) \leq E(\tuple{u',v,\varrho}, \tuple{u,u'}) \fto M(\tuple{u',v,\varrho}) \leq (0.7 \fto 0.5) = 0.5. \label{eq: JUZJA 2}
\end{eqnarray*}
These inequalities, together with~\eqref{eq: JHDSD 1}, imply that 
\begin{equation}\label{eq: JHFJS}
M|_{\Vmin} \leq \{\tuple{v,v'}\!:\!0.5, \tuple{w,v'}\!:\!0.5, \tuple{u,u'}\!:\!0.5\}.
\end{equation}
By~\eqref{eq: JHDSD 3}, it follows that
\begin{align}
M|_{\Vmax} \leq \{ & \tuple{u,v',\varrho}\!:\!0.5, \tuple{v,v',\varrho}\!:\!0.5, \tuple{v',v,\varrho}\!:\!0.5, \tuple{w,v',\varrho}\!:\!0.5, \nonumber \\ 
                   & \tuple{v',w,\varrho}\!:\!0.5, \tuple{u',v,\varrho}\!:\!0.5, \tuple{u',w,\varrho}\!:\!0.5 \}. \label{eq: JHRJA}
\end{align}
It is straightforward to verify that the union of the fuzzy sets on the right hand sides of the inequalities~\eqref{eq: JHFJS} and~\eqref{eq: JHRJA} is a correct marking of~$\mN$ (since it satisfies~\eqref{eq: JHDSD 4} and~\eqref{eq: JHDSD 5}). This, together with~\eqref{eq: JHFJS}, implies that 
\begin{equation*}\label{eq: JHFJS 2}
M|_{\Vmin} =  \{\tuple{v,v'}\!:\!0.5, \tuple{w,v'}\!:\!0.5, \tuple{u,u'}\!:\!0.5\}.
\end{equation*}
By Theorem~\ref{theorem: HDFMX}, it follows that $Z = M|_{\Vmin} =  \{\tuple{v,v'}\!:\!0.5, \tuple{w,v'}\!:\!0.5, \tuple{u,u'}\!:\!0.5\}$ is the greatest fuzzy directed simulation between~$\mM$ and~$\mMp$. This is consistent with Example~\ref{example: JHRJA}.
\myend
\end{example}

\section{Implementation and performance tests}
\label{section: impl and performance}

In view of Theorem~\ref{theorem: HDFMX}, to compute the greatest fuzzy directed simulation between two finite fuzzy Kripke models $\mM$ and $\mMp$, we can proceed as follows. First, construct the minimax net $\mN = \tuple{\Vmin, \Vmax, E, \Mz}$ corresponding to $\tuple{\mM,\mM'}$ via directed simulation. 
\revisedA{(Here, one can discard the assumption that $\mM$ and $\mMp$ have disjoint domains and apply the modifications mentioned in Remark~\ref{remark: JHZHA}.)} 
Next, compute the greatest correct marking $M$ of $\mN$, and finally return the restriction $M|_{\Vmin}$. The computation of this marking is the central step. 
When the underlying residuated lattice is the product or \L{}ukasiewicz structure, this computation can be carried out using the algorithms from~\cite{NguyenMS23-supplement}. These algorithms have exponential time complexity in general (when $\mN$ is cyclic), but incorporate important optimizations. 
In contrast, when the underlying residuated lattice is the G\"odel structure, the computation of the greatest correct marking can be performed using an algorithm with a polynomial time complexity~\cite{CompMinimaxG}. 

Based on this approach, we have developed a Python package~\cite{FDSML-prog} for computing the greatest fuzzy directed simulation between two finite fuzzy Kripke models over the G\"odel, product, and \L{}ukasiewicz structures. The package implements the aforementioned algorithms from~\cite{NguyenMS23-supplement} and incorporates the previously developed implementation of the aforementioned algorithm from~\cite{CompMinimaxG}. The implementation of the algorithms from~\cite{NguyenMS23-supplement} uses a tolerance parameter $\varepsilon$, whose default value is $10^{-12}$. Two fuzzy values whose difference is smaller than $\varepsilon$ are treated as equal. \revisedA{This tolerance parameter can easily be changed by the user, for example, to~$0$.}

The package provides, in the module {\em experiments1.py}, the function
{\small\markRevisedA 
\begin{verbatim}
    computeGreatestBisimulationOrSimulationOrDirectedSimulation(
        G, G', purpose, structure, precision)
\end{verbatim}
}

\noindent
which computes the greatest fuzzy bisimulation, simulation, or directed simulation between two finite fuzzy graphs $G$ and $G'$, depending on the value of the parameter {\em purpose}, which belongs to the set
\{``bisimulation'', ``simulation'', ``directed simulation''\}. The parameter {\em structure} specifies the underlying residuated lattice and takes one of the values `G' (G\"odel), `P' (product), or `L' (\L{}ukasiewicz). \revisedA{The parameter {\em precision} specifies the tolerance parameter $\varepsilon$ mentioned above. It is used only when {\em structure} is `P' or `L', and its default value is $10^{-12}$.} 
The function returns a triple $(Z,t_1,t_2)$, where $Z$ is a Python dictionary representing the computed fuzzy relation, $t_1$ is the time (in seconds) required to construct the minimax net corresponding to $\tuple{G,G'}$, and $t_2$ is the time required to compute its greatest correct marking. Consequently, $t_1+t_2$ is the total running time of the computation.

Fuzzy Kripke models are represented as fuzzy graphs and can be stored in text files. For example, the fuzzy Kripke model $\mM$ from Example~\ref{example: JHRJA} is stored in the file ``g5.in'' included in~\cite{FDSML-prog}, whose contents are

{\small
\begin{verbatim}
    u p 1
    v q 0.7
    w q 0.4

    u v r 0.7
    u w r 0.9
    v w r 0.6
    w v r 0.8
\end{verbatim}
}

\noindent
The model can be loaded into memory by
{\small 
\begin{verbatim}
    G = RawFuzzyGraph.fromFile("g5.in")
\end{verbatim}
}

Most details of Examples~\ref{example: JHRJA}, \ref{example: JHIUS}, and~\ref{example: JHRJS} can be reproduced by running the modules {\em CompMinimaxG.py} and {\em CompMinimaxPL.py} from~\cite{FDSML-prog} and examining the output corresponding to the input files ``g5.in'' and ``g5p.in''. Additional information can be displayed by modifying the scripts appropriately.

\begin{figure}
\includegraphics[scale=0.66]{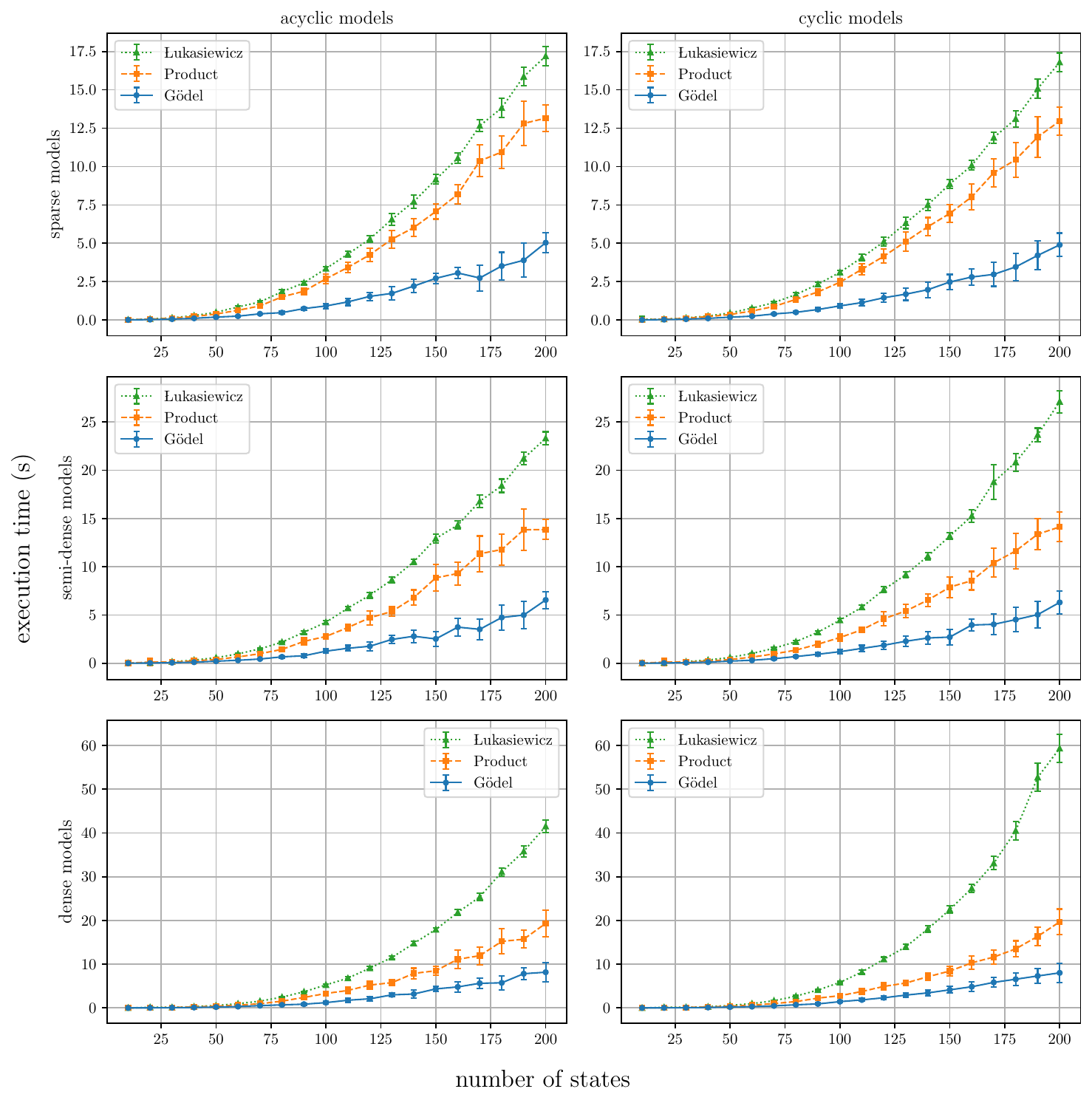}
\caption{Results of the performance tests -- Part I\label{fig: JHDJW1}}
\end{figure}

\begin{figure}
\centering
\resizebox{0.66\width}{0.5\height}{\includegraphics{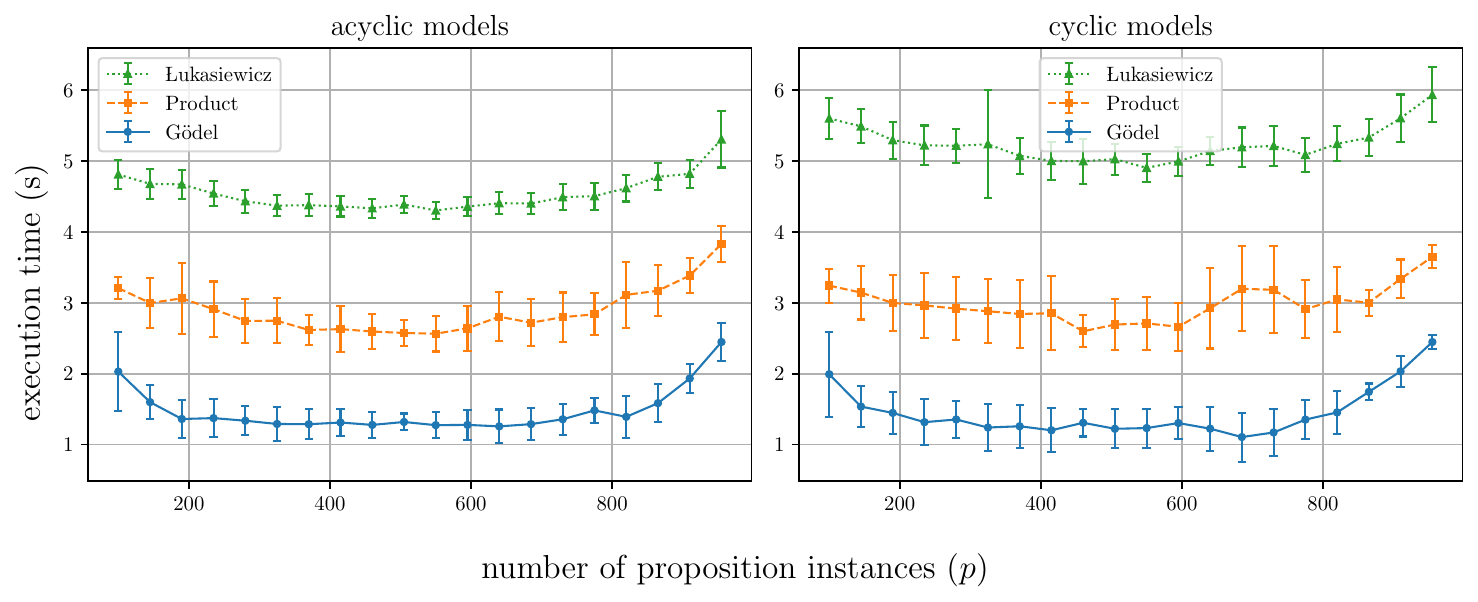}}
\smallskip

\resizebox{0.66\width}{0.45\height}{\includegraphics{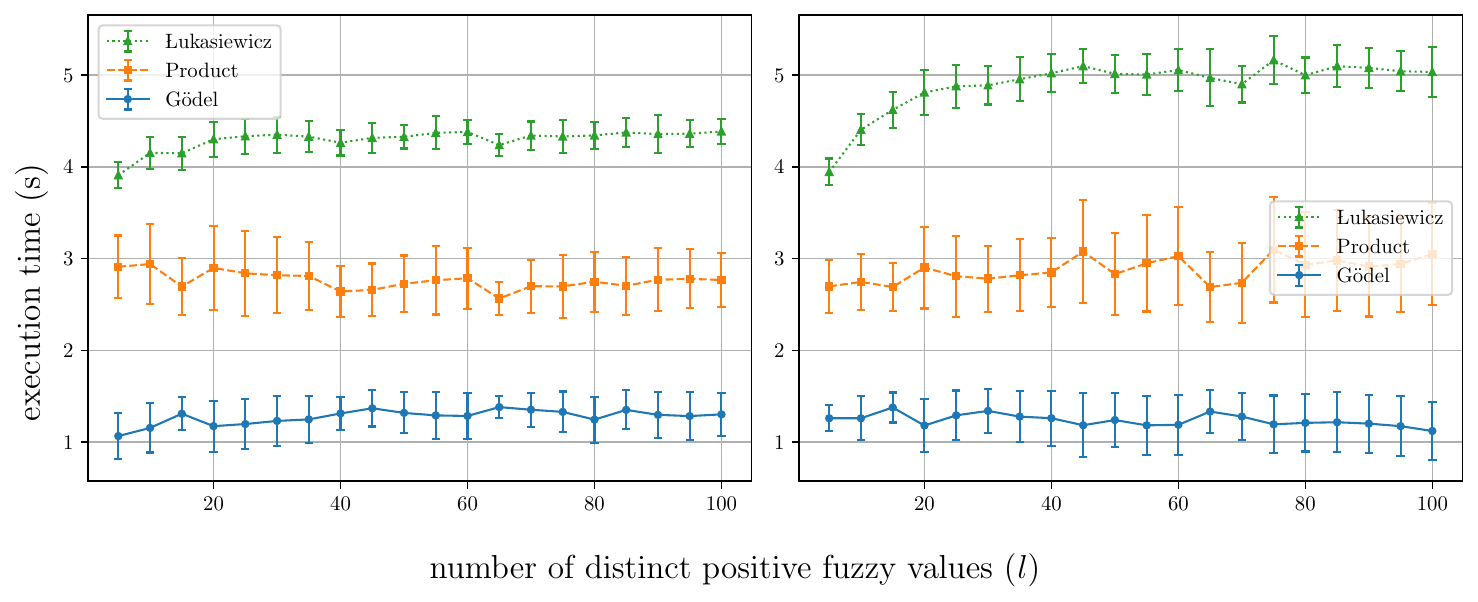}}
\smallskip

\resizebox{0.66\width}{0.45\height}{\includegraphics{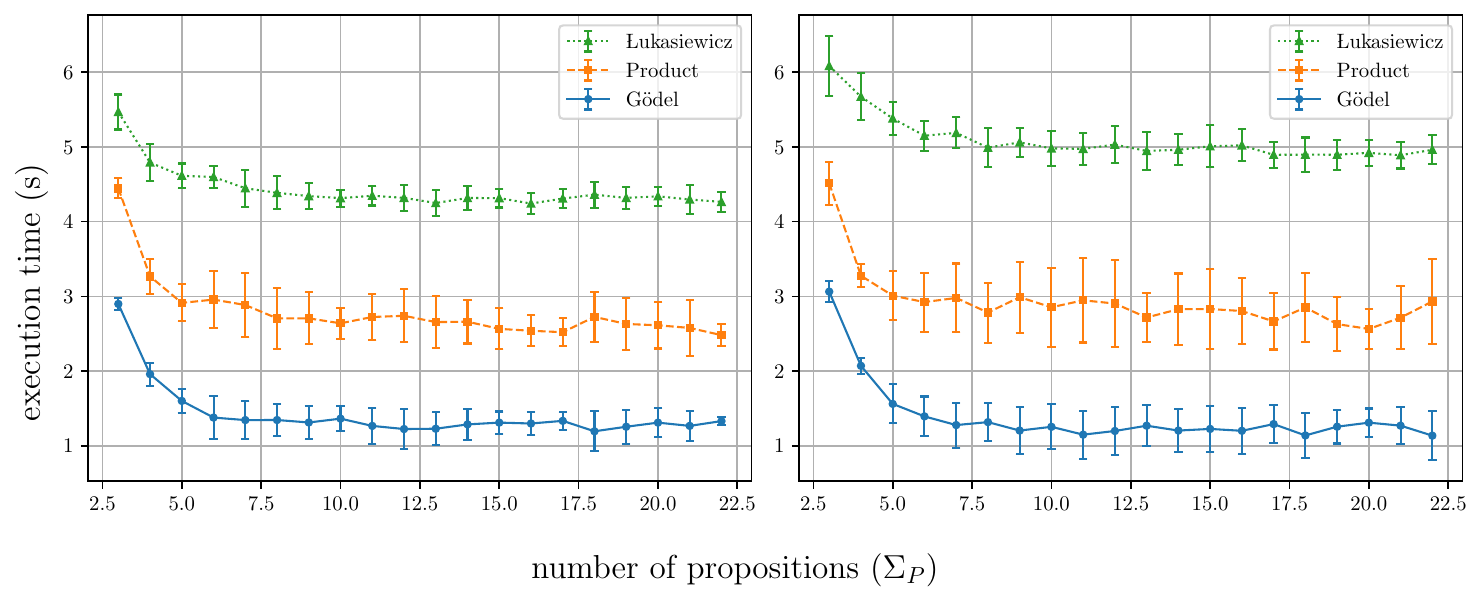}}
\smallskip

\resizebox{0.66\width}{0.45\height}{\includegraphics{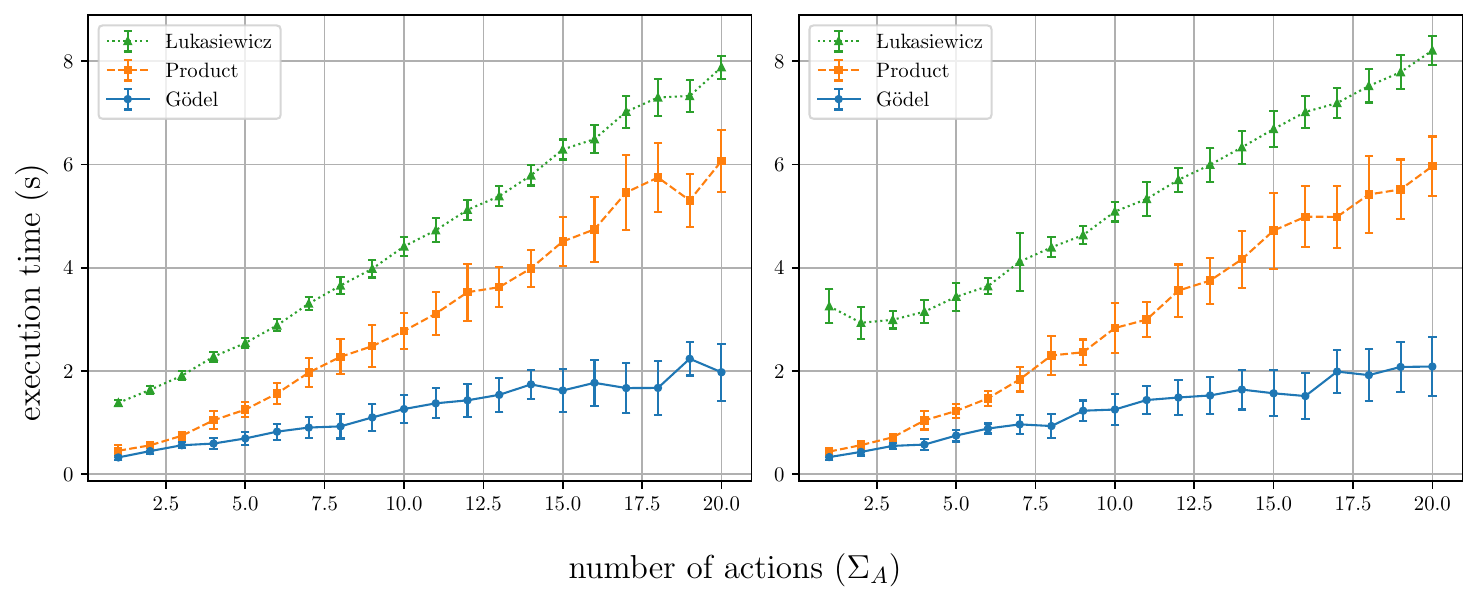}}
\caption{Results of the performance tests -- Part II\label{fig: JHDJW2}}
\end{figure}

We evaluated the performance of our implementation for computing the greatest fuzzy directed simulation between two finite fuzzy Kripke models over the G\"odel, product, and \L{}ukasiewicz structures. The experiments use randomly generated fuzzy Kripke models characterized by the following parameters:
\begin{itemize}
\item $n$: the number of states, $|\DeltaM|$ (default value: $100$);
\item $m$: the number of transitions,
$\#\{\tuple{x,\varrho,y}\in\DeltaM\times\SA\times\DeltaM\mid\varrho^\mM(x,y)>0\}$
(default value: $800$);
\item $p$: the number of non-zero proposition instances,
$\#\{\tuple{p,x}\in\SP\times\DeltaM\mid p^\mM(x)>0\}$
(default value: $3n$);
\item $l$: the number of distinct positive fuzzy values used in $\mM$ (default value: $100$);
\item $|\SP|$: the number of propositions (default value: $10$);
\item $|\SA|$: the number of actions (default value: $10$);
\item $\mathit{acyclic}$: whether $\mM$ is acyclic \revisedA{(disregarding self-loops)}.
\end{itemize}
\revisedA{The random instances were generated using uniform distributions. For example, a transition $\tuple{x,\varrho,y}$ is generated by drawing $x$ and $y$ uniformly at random from the set of states, drawing $\varrho$ uniformly at random from $\SA$, and choosing $\varrho^\mM(x,y)$ uniformly at random from the set $\{i/l \mid$ $i \in \NN$, $0 \leq i \leq l\}$. For acyclic models, we exchange $x$ and $y$ whenever $x > y$, assuming that the set of states is linearly ordered (in fact, we have used integers to identify states).}

All experiments were performed using a standard user account on a virtualized Linux server running Debian GNU/Linux~12 (Bookworm). The server was equipped with 64 logical Intel Broadwell CPU cores operating at approximately 2.2~GHz and 184~GiB of RAM. The user account was subject to standard resource controls. 

Figures~\ref{fig: JHDJW1} and~\ref{fig: JHDJW2} summarize the experimental results. In both figures, the left and right columns correspond to acyclic and cyclic models, respectively, and the vertical axis represents the running time in seconds. 
In Figure~\ref{fig: JHDJW1}, the horizontal axis represents the number of states ($n$) in each of the two input models, and the three rows correspond to sparse ($m=3n$), semi-dense ($m=n\log n$), and dense ($m=n^2/10$) models, respectively. 
In Figure~\ref{fig: JHDJW2}, the horizontal axis in the subsequent rows represents, respectively, the number of proposition instances ($p$), the number of distinct positive fuzzy values ($l$), the number of propositions ($|\SP|$), and the number of actions ($|\SA|$). Thus, each group of experiments studies the influence of a single parameter while all remaining parameters are kept at their default values. The only exception is that, when $n$ varies, $m$ is adjusted according to the three density settings.
Each experiment was repeated \revisedA{30} times, and the reported running time is the \revisedA{mean, together with the corresponding standard deviation,} over these repetitions. 

The experimental results lead to the following observations:
\begin{itemize}
\item The proposed method scales well with respect to the parameters $p$, $l$, $|\SP|$, and $|\SA|$ \revisedA{in our experimental setting}.

\item Although the worst-case time complexity is exponential for cyclic models over the product and \L{}ukasiewicz structures, the optimization techniques incorporated into the algorithms for computing the greatest correct marking of a fuzzy minimax net~\cite{NguyenMS23-supplement} provide substantial practical speedups. In particular, for the product structure, the running times for cyclic models are comparable to those for acyclic models. This behavior is due, in part, to the use of the tolerance parameter~$\varepsilon$.
\end{itemize}

The performance tests can be reproduced by running the modules {\em experiments1.py} and {\em process$\_$results1.py} (in this order) of~\cite{FDSML-prog}. 
\revisedA{The former generates a file {\em results1.txt}, whereas the latter processes this file and produces plots. Such plots have been used to create Figures~\ref{fig: JHDJW1} and~\ref{fig: JHDJW2}.}
 
\section{Conclusions}
\label{section: conc}

{\markRevisedA
We have introduced fuzzy directed simulations between fuzzy Kripke models over linear and complete residuated lattices and established their logical characterizations. All results presented in Section~\ref{section: fbir}, including Theorem~\ref{theorem: preservation} (on preservation of positive formulas under fuzzy directed simulations) and Theorem~\ref{theorem: HM} (on the Hennessy--Milner property for fuzzy directed simulations), are new. Some parts of their proofs are obtained by adapting the proofs from~\cite{FBSML}, which concern fuzzy bisimulations rather than fuzzy directed simulations. However, the part concerning condition~\eqref{eq: FDS3} in the proof of Theorem~\ref{theorem: HM} is entirely new and is crucial for handling positive formulas and fuzzy directed simulations.

We have also presented the first method for computing the greatest fuzzy directed simulation between two finite fuzzy Kripke models. It is based on a reduction to the problem of computing the greatest correct marking of a finite fuzzy minimax net~\cite{DBLP:journals/tfs/NguyenMS23}. The method has been used and acknowledged in~\cite{CompMinimaxG} when specialized to the case where the underlying residuated lattice is the G\"odel structure, in order to obtain an efficient algorithm for this case.

Furthermore, we have implemented that method for the case where the underlying residuated lattice is the G\"odel, product, or \L{}ukasiewicz structure, experimentally evaluated its performance, and presented the obtained results.
} 

This work fills the existing gap in the picture of bisimulations, simulations, and directed simulations for fuzzy systems. As future work, we intend to exploit crisp and fuzzy directed simulations to study the minimization of fuzzy Kripke models while preserving positive modal formulas.

{\markRevisedA
\section*{Acknowledgments}

We would like to thank an anonymous reviewer for useful comments, which have helped us improve the quality of this paper.
}


\biboptions{sort&compress}
\bibliography{BSfDL}
\bibliographystyle{elsarticle-harv}

\appendix

\section{Proofs}
\label{section: appendix}

In this appendix, we present the proofs of the results from Sections~\ref{section: fbir} and~\ref{section: computation}. Most of them are obtained by adapting the proofs from~\cite{FBSML,DBLP:journals/tfs/NguyenMS23}, which deal with fuzzy bisimulations rather than fuzzy directed simulations. The exceptions are the proofs of Proposition~\ref{prop: HFHSJ} and Theorem~\ref{theorem: HM}. While the adaptations involve nontrivial modifications, the part concerning condition~\eqref{eq: FDS3} in the proof of Theorem~\ref{theorem: HM} is entirely new and is crucial for handling positive formulas and fuzzy directed simulations.

\ProofPropositionHFHS

\ProofLemmaGDHAW

\ProofTheoremHM

\ProofLemmaHDJHA

\ProofTheoremHDFMX


\end{document}